\documentclass[11pt]{article}
\usepackage{fullpage}
\usepackage[T1]{fontenc}
\usepackage[utf8]{inputenc}
\usepackage{xcolor}
\usepackage[colorlinks=true]{hyperref}
\definecolor{purple}{HTML}{7068A4}
\definecolor{blue}{HTML}{3F70B2}
\definecolor{orange}{HTML}{A47458}
\definecolor{yellow}{HTML}{9B7A3C}
\hypersetup{
    linkcolor=blue
    ,citecolor=orange
    ,filecolor=purple
    ,urlcolor=orange
    ,menucolor=purple
    ,runcolor=purple
}
\usepackage{amsmath, amssymb, amsthm, amsfonts, graphicx, color, subcaption, enumerate, bm, cleveref, enumitem, array, mathtools}
\usepackage[colorinlistoftodos,prependcaption,textsize=tiny]{todonotes}

\usepackage{multicol, multirow}
\usepackage{thmtools, thm-restate}
\usepackage{ragged2e}
\usepackage{lineno}
\usepackage{framed}
\usepackage[framemethod=tikz]{mdframed}

\usepackage[ruled,linesnumbered]{algorithm2e}
\SetKwComment{Comment}{/* }{ */}

\usepackage[style=trad-alpha,natbib=true,maxcitenames=4]{biblatex}
\usepackage{subfiles}
\graphicspath{{graphics/}}

\crefname{claim}{Claim}{Claims}
\crefname{property}{Property}{Properties}
\crefname{algocf}{Algorithm}{Algorithms}
\Crefname{algocf}{Algorithm}{Algorithms}

\usepackage{soul}
\usepackage{enumerate}

\usepackage[style=trad-alpha,natbib=true,maxcitenames=4]{biblatex}
\makeatletter
\g@addto@macro\bfseries{\boldmath}
\makeatother

\newtheorem{question}{Question}
\newtheorem{lemma}{Lemma}[section]

\newtheorem{corollary}[lemma]{Corollary}
\newtheorem{observation}[lemma]{Observation}
\newtheorem{proposition}[lemma]{Proposition}

\theoremstyle{definition}
\newtheorem{definition}[lemma]{Definition}

\theoremstyle{remark}

\newtheorem*{remark*}{Remark}

\renewcommand{\epsilon}{\varepsilon}

\renewcommand{\emptyset}{\varnothing}

\newcommand{\dist}{\mathsf{dist}}
\newcommand{\acc}{\mathsf{accept}}
\newcommand{\rej}{\mathsf{reject}}
\newcommand{\true}{\mathsf{true}}
\newcommand{\false}{\mathsf{false}}
\newcommand{\MSO}{\mathsf{MSO}_2}
\newcommand{\ID}{\mathsf{ID}}
\newcommand{\inc}{\mathsf{inc}}
\newcommand{\adj}{\mathsf{adj}}

\newcommand{\vst}{v_{\mathsf{st}}}
\newcommand{\ved}{v_{\mathsf{ed}}}

\newcommand{\tin}{\tau^{\mathsf{in}}}
\newcommand{\tout}{\tau^{\mathsf{out}}}

\newcommand{\fin}{\phi^{\mathsf{in}}}
\newcommand{\fout}{\phi^{\mathsf{out}}}

\newcommand{\insv}{\mathsf{V}\text{-}\mathsf{insert}}
\newcommand{\inse}{\mathsf{E}\text{-}\mathsf{insert}}

\newcommand{\tmerge}{\mathsf{Tree}\text{-}\mathsf{merge}}
\newcommand{\pmerge}{\mathsf{Parent}\text{-}\mathsf{merge}}
\newcommand{\bmerge}{\mathsf{Bridge}\text{-}\mathsf{merge}}

\newcommand{\fb}{f_{\mathbf{B}}}
\newcommand{\fp}{f_{\mathbf{P}}}

\newcommand{\vnode}{\mathsf{V}\text{-}\mathsf{node}}
\newcommand{\enode}{\mathsf{E}\text{-}\mathsf{node}}
\newcommand{\bnode}{\mathsf{B}\text{-}\mathsf{node}}
\newcommand{\pnode}{\mathsf{P}\text{-}\mathsf{node}}
\newcommand{\tnode}{\mathsf{T}\text{-}\mathsf{node}}

\title{Optimal local certification on graphs of bounded pathwidth}

\author{Dan Alden Baterisna\footnote{National University of Singapore. ORCID: 0009-0002-8102-2895. Email: danbaterisna@u.nus.edu} \and Yi-Jun Chang\footnote{National University of Singapore. ORCID: 0000-0002-0109-2432. Email: cyijun@nus.edu.sg}}

\date{}

\begin{document}

\maketitle
\thispagestyle{empty}

\begin{abstract}
We present proof labeling schemes for graphs with bounded pathwidth that can decide any graph property expressible in monadic second-order ($\MSO$) logic using $O(\log n)$-bit vertex labels. Examples of such properties include planarity, Hamiltonicity, $k$-colorability, $H$-minor-freeness, admitting a perfect matching, and having a vertex cover of a given size.

Our proof labeling schemes improve upon a recent result by Fraigniaud, Montealegre, Rapaport, and Todinca (Algorithmica 2024), which achieved the same result for graphs of bounded treewidth but required $O(\log^2 n)$-bit labels. Our improved label size $O(\log n)$ is \emph{optimal}, as it is well-known that any proof labeling scheme that accepts paths and rejects cycles requires labels of size $\Omega(\log n)$.

Our result implies that graphs with pathwidth at most $k$ can be certified using $O(\log n)$-bit labels for any fixed constant $k$. Applying the Excluding Forest Theorem of Robertson and Seymour, we deduce that the class of $F$-minor-free graphs can be certified with $O(\log n)$-bit labels for any fixed forest $F$, thereby providing an affirmative answer to an open question posed by Bousquet, Feuilloley, and Pierron (Journal of Parallel and Distributed Computing 2024).
\end{abstract}

\newpage
\bigskip
\tableofcontents
\bigskip
\thispagestyle{empty}



\newpage
\pagenumbering{arabic}

\section{Introduction}\label{sect:intro}

Local certification, introduced by Korman, Kutten, and Peleg~\cite{korman03032010} as \emph{proof labeling schemes}, refers to the task of locally verifying that a network satisfies a given global property. The study of local certification originated from the field of \emph{self-stabilization}, where each processor in the network must be able to efficiently verify the validity of a solution based on local information, independent of the stabilization strategy in use. Local certification is also important in \emph{fault-tolerance}, as it allows each processor to detect if the network is in an invalid state, enabling the system to respond and address potential errors. For more about local certification, refer to the survey by Feuilloley~\cite{dmtcs:8479}.

\subsection{Proof labeling schemes} We model a communication network as an $n$-vertex connected undirected graph $G=(V,E)$, where each vertex $v \in V$ represents a processor and each edge $e=\{u,v\}\in E$ represents a communication link between the two endpoints $u$ and $v$. A \emph{state assignment function}  $s : V \to S$, for some state space $S$, is used to represent the state of each processor in the network. The state assignment function can be used to encode any input given to the processors. Unless otherwise stated, we assume that the state $s(v)$ of each vertex $v \in V$ includes an $O(\log n)$-bit distinct identifier $\ID(v)$. 
A graph $G$, together with a state assignment function $s$, is called a \emph{configuration} $(G,s)$.

Given a family of configurations $\mathcal{F}$ and a boolean predicate $\mathcal{P}$ over $\mathcal{F}$, a \emph{proof labeling scheme} for $(\mathcal{F}, \mathcal{P})$ is a pair of two algorithms $\mathbf{P}$ and $\mathbf{V}$, as follows.
\begin{description}
    \item[Certificate assignment:] $\mathbf{P}$ is a centralized algorithm that takes in a configuration $(G,s) \in \mathcal{F}$ such that $\mathcal{P}(G,s)=\true$, and assigns a label $\ell(v)$ to every vertex $v$ in the graph $G$. Unless otherwise specified, $\mathbf{P}$ may use unlimited computational resources.
    \item[Local verification:] $\mathbf{V}$ is a distributed algorithm that runs on every vertex $v$ in the network $G$, where $v$ initially knows its state $s(v)$ and its label $\ell(v)$. For a vertex $v$, the input consists of its state $s(v)$, its label $\ell(v)$, and the multiset of labels $\{\ell(u) \, | \, u \in N(v)\}$ of its neighbors $N(v)$. The output of the algorithm for a vertex $v$ is either $\acc$ or $\rej$. If the output is $\acc$ for all vertices, then we say that $\mathbf{V}$ accepts $(G,s)$. Otherwise, we say that  $\mathbf{V}$ rejects $(G,s)$.
\end{description}
A proof labeling scheme $(\mathbf{P}, \mathbf{V})$ is correct for $(\mathcal{F}, \mathcal{P})$ if the following requirements are met.
\begin{description}
    \item[Completeness:] If $\mathcal{P}(G,s)=\true$, then the labeling $\ell$ computed by $\mathbf{P}$ makes $\mathbf{V}$ accepts $(G,s)$.
    \item[Soundness:] If $\mathcal{P}(G,s)=\false$, then  $\mathbf{V}$ rejects $(G,s)$ regardless of the given labeling $\ell$. 
\end{description}

Intuitively, $\mathbf{P}$ is a centralized prover that assigns a label $\ell(v)$ to each vertex $v$. These labels are then exchanged along the edges of the graph in one round of communication. After that, the correctness of the proof can be checked locally using the verification algorithm  $\mathbf{V}$.

For the rest of the paper, we often omit explicitly writing the state assignment function $s$, and treating $\mathcal{F}$ as a class of graphs that may have input labels on vertices and edges, where $s(v)$ includes not only the unique identifier $\ID(v)$ but also the input labels of vertex $v$ and all edges incident to $v$. The unique identifiers are not considered part of the input labels.

\subparagraph{Complexity measure.} The main complexity measure of a proof labeling scheme is the \emph{proof size} -- the maximum length of a label that $\mathbf{P}$ assigns to a vertex, in terms of the number $n = |V|$ of vertices in the graph $G=(V,E)$.
 For instance, if we wish to certify that the graph is bipartite, then one bit suffices -- $\mathbf{P}$ labels each vertex $v$ using one bit to encode the color of $v$ in a proper 2-coloring, and $\mathbf{V}$ checks if $\ell(v) \neq \ell(u)$ for all neighbors $u \in N(v)$. In some other cases, we are not so lucky: There are graph properties, such as certifying that the graph is symmetrical, where $\Omega(n^2)$ bits are necessary~\cite{korman03032010,v012a019}, meaning there is no asymptotically better method than encoding the entire graph in each label.

 \subsection{Monadic second-order logic}
 The study of proof labeling schemes aims to answer the question: For which natural graph properties $\mathcal{P}$ and families of graphs $\mathcal{F}$ do good proof labeling schemes exist? One common way of expressing graph properties is to write them as logical predicates in a certain language. A language often used for this purpose is the fragment of \emph{monadic second-order logic} ($\MSO$) in graphs, where four types of variables are allowed: vertices, edges, vertex sets, and edge sets. The quantifiers $\forall$ and $\exists$ can be applied to any variables. In addition to the basic logical connectives $\{\neg,\vee,\wedge,\rightarrow,\leftrightarrow\}$, the following binary predicates are allowed.
    \begin{itemize}
    \item $v \in U$, for a vertex variable $v$ and vertex set variable $U$.
    \item $e \in F$, for an edge variable $e$ and edge set variable $F$.
    \item $\inc(e, v)$, for an edge variable $e$ and vertex variable $v$; this is interpreted as $e$ being an edge incident to $v$.
    \item $\adj(u, v)$, for vertex variables $u$ and $v$; this is interpreted as $u$ being adjacent to $v$.
    \item Equality for vertices, edges, vertex sets, and edge sets.
    \end{itemize}

Many common graph properties, including planarity, Hamiltonicity, $k$-colorability, $H$-minor-freeness, admitting a perfect matching, and having a vertex cover of a given size, can be written as an $\MSO$ property. Refer to~\citeauthor{borie1992automatic}~\cite{borie1992automatic} for a list of such properties.

\subparagraph{Algorithmic meta-theorems.}
Monadic second-order logic plays a crucial role in the field of \emph{algorithmic meta-theorems}, which aims to develop algorithmic results that apply to large families of computational problems rather than just specific problems.
One of the most prominent algorithmic meta-theorems is \emph{Courcelle's theorem}, which asserts that any property expressible in $\MSO$ is decidable in linear time for graphs of bounded \emph{treewidth}. This result was first demonstrated by Courcelle~\cite{COURCELLE199012} and later independently rediscovered by Borie, Parker, and Tovey~\cite{borie1992automatic}. Additionally, Courcelle's theorem extends to graphs where vertices and edges are labeled with inputs from a fixed finite set, by including predicates that describe the labels, or by representing the labels using unquantified vertex set or edge set variables~\cite{arnborg1991easy}.

\subparagraph{Local certification.}
Fraigniaud, Montealegre, Rapaport, and Todinca~\cite{fraigniaud2024meta}  established a result analogous to Courcelle's theorem in the context of local certification: Any $\MSO$ property in bounded-treewidth graphs can be certified using an $O(\log^2 n)$-bit proof labeling scheme. Their work raises an intriguing question: Can the label size be further reduced to $O(\log n)$? Currently, the best known lower bound for this problem is $\Omega(\log n)$, because any proof labeling scheme that accepts paths and rejects cycles necessitates labels of size $\Omega(\log n)$~\cite{korman03032010}. Observe that both paths and cycles have bounded treewidth and are also $\MSO$ properties. The label size $O(\log n)$ is generally seen as the ``gold standard'' in local certification, as very little can be done with labels of size $o(\log n)$. Indeed, even storing the vertex identifiers requires $\Omega(\log n)$ bits.


\subsection{Minor-closed graph properties}
A graph $H$ is a \emph{minor} of $G$ if $H$ can be obtained from $G$ through a series of vertex deletions, edge deletions, and edge contractions. 
 A graph property $\mathcal{P}$ is said to be \emph{minor-closed} if $\mathcal{P}(G)=\true$ implies that  $\mathcal{P}(H)=\true$ for all minors $H$ of $G$. 
  A cornerstone result in structural graph theory is the \emph{graph minor theorem} of Robertson and Seymour~\cite{ROBERTSON2004325}, which establishes that \emph{any}  minor-closed graph property $\mathcal{P}$ can be characterized by forbidden minors: There is a \emph{finite} list of graphs such that $\mathcal{P}(G)=\true$  if and only if $G$ does not contain any graph from this list as a minor. Consequently, any minor-closed graph property can be expressed in $\MSO$, as the property of containing a fixed graph $H$ as a minor is $\MSO$.

Many graph classes, such as forests, cacti, bounded-genus graphs, and bounded-treewidth graphs, are minor-closed and thus can be characterized by forbidden minors. In local certification, substantial effort has been devoted to designing proof labeling schemes for various minor-closed properties, with the ultimate goal of resolving the following problem.

\medskip
\centerline{%
    \parbox{0.7\linewidth}{
        \begin{mdframed}[hidealllines=true,backgroundcolor=gray!25]\begin{center}
        \vspace{-0.3cm}
        \begin{question}[{\cite[Open question 4]{dmtcs:8479}}]
            Can we certify all minor-closed graph properties by $O(\log n)$-bit proof labeling schemes?
        \end{question}
            \end{center}\end{mdframed}
    }%
}
\medskip

As an application of a more general method, it was demonstrated by~\citeauthor{naor2020power}~\cite{naor2020power} that planar graphs can be certified using $O(\log n)$ bits in a more powerful model of distributed interactive proofs. Subsequently, a $O(\log n)$-bit proof labeling scheme was designed for planarity in the standard model~\cite{feuilloley2021compact}, which was then extended to cover all bounded-genus graphs~\cite{ESPERET202268,feuilloley2023local}.

The result of Fraigniaud, Montealegre, Rapaport, and Todinca~\cite{fraigniaud2024meta} discussed above implies that graphs with treewidth at most $k$ can be locally certified using $O(\log^2 n)$-bit labels for any fixed constant $k$. Combining this result with the Excluding Grid Theorem of Robertson and Seymour~\cite{ROBERTSON198692}, it follows that the class of $H$-minor-free graphs can be locally certified with $O(\log^2 n)$-bit labels for any fixed planar graph $H$. This naturally leads to the question: For which graph $H$, can the class of $H$-minor-free graphs be locally certified with $O(\log n)$-bit labels? In a recent work, \citeauthor{bousquet2024local} explored local certification of $H$-minor-free graphs for some specific small minors $H$~\cite{bousquet2024local}, showing that for $H \in \{K_3, K_4, K_{2,3}, K_{2,4}, \text{diamond}\}$, the class of $H$-minor-free graphs admits a proof labeling scheme with optimal $O(\log n)$-bit labels. To continue this line of research, the following appears to be the simplest open question to tackle.
 
\medskip
\centerline{%
    \parbox{0.7\linewidth}{
        \begin{mdframed}[hidealllines=true,backgroundcolor=gray!25]\begin{center}
                \vspace{-0.3cm}
        \begin{question}[{\cite[Question 54]{bousquet2024local}}]
            Can we certify $T$-minor-free graphs by $O(\log n)$-bit proof labeling schemes for any tree $T$?\label{q2}
        \end{question}
       
            \end{center}\end{mdframed}
    }%
}
\medskip

\subsection{Our results} 

We answer \Cref{q2} affirmatively by studying local certification on graphs of bounded \emph{pathwidth}.

\begin{definition}[{Path decomposition~\cite{ROBERTSON198339}}]
    \label{pathwidth-defn}
    A \underline{path decomposition} of a graph $G=(V,E)$ is a sequence of vertex subsets $(X_1, X_2, \ldots, X_s)$ of $G$ satisfying the following conditions.
\begin{description}
    \item[(P1)] For each edge $e=\{u,v\} \in E$, there exists an index $i$ such that $\{u,v\} \subseteq X_i$.
    \item[(P2)] For every three indices $i \leq j \leq k$, $X_i \cap X_k \subseteq X_j$.
\end{description}  
    The \underline{width} of a path decomposition is $\max_{i} |X_i|-1$. The \underline{pathwidth} of a graph $G$ is the minimum width of any path decomposition of $G$.
\end{definition}

Observe that (P2) implies that for each vertex $v \in V$, there exists an interval $I_v = [L_v, R_v]$ such that $v \in X_i$ if and only if $i \in I_v$. Therefore, a path decomposition can be seen as an assignment of an interval $I_v = [L_v, R_v]$ to each vertex $v \in V$ such that $I_u \cap I_v \neq \emptyset$ for each edge $e=\{u,v\} \in E$, and the width of the decomposition is the maximum number of intervals with a non-empty intersection minus one. See \Cref{fig:f1} for an illustration.

\begin{figure}[ht!]
    \centering
    \includegraphics[scale=0.65]{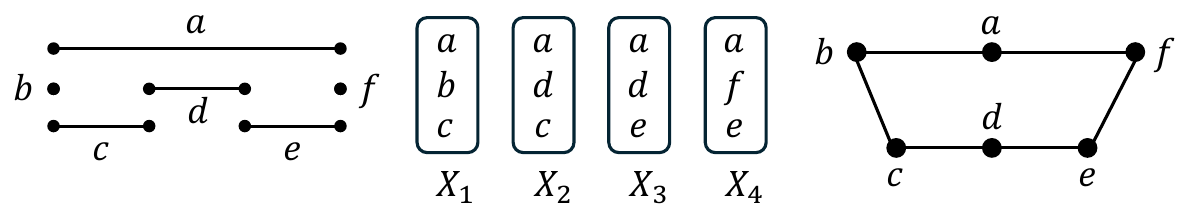}
    \caption{Path decomposition and interval representation of a $6$-cycle.}
    \label{fig:f1}
\end{figure}

In a sense, the pathwidth of a graph measures how close it is to being a path: the lower it is, the more path-like it is. Indeed, paths have pathwidth one. Similar measures of graph density include treewidth, cutwidth, and cliquewidth. There is a rich theory of computational complexity surrounding these graph measures: If we limit our attention to graphs for which these measures are upper-bounded by some constant, then some problems that are NP-hard in general become solvable in polynomial time or even in linear time. For more on the field of fixed-parameter tractability, particularly these graph measures, refer to the book by~\citeauthor{downey27082016}~\cite{downey27082016}.

\subparagraph{Our contribution.} The main contribution of this work is to show that for any $\MSO$ graph property $\phi$ and any graph class $\mathcal{G}$ where the pathwidth is bounded, there exists an $O(\log n)$-bit proof labeling scheme for $\phi$ in the class of graphs $\mathcal{G}$, improving upon the prior work by Fraigniaud, Montealegre, Rapaport, and Todinca~\cite{fraigniaud2024meta}  which requires $O(\log^2 n)$-bit labels. More precisely, we prove the following theorem.

\begin{restatable}{theorem}{mainthm}\label{thm:main}
For any integer $k \geq 1$, for any $\MSO$ graph property $\phi$, there exists an $O(\log n)$-bit proof labeling scheme for the graph property $\phi \wedge (\operatorname{pathwidth}\leq k)$. 
\end{restatable}


Since $(\operatorname{pathwidth}\leq k)$ itself is an $\MSO$ graph property, by setting $\phi = (\operatorname{pathwidth}\leq k)$,  \Cref{thm:main} allows us to locally certify $(\operatorname{pathwidth}\leq k)$ with $O(\log n)$ bits.
Combining \Cref{thm:main} with the Excluding Forest Theorem of \citeauthor{ROBERTSON198339}, we obtain the following result, affirmatively answering~\cite[Question 54]{bousquet2024local}.

\begin{corollary}
For any forest $F$, there is an $O(\log n)$-bit proof labeling scheme to certify the class of $F$-minor-free graphs.
\end{corollary}
\begin{proof}
The Excluding Forest Theorem of \citeauthor{ROBERTSON198339} states that for every forest $F$, there exists a number $k$ such that the pathwidth of every $F$-minor-free graph is at most $k$~\cite{ROBERTSON198339}. Therefore, the proof labeling scheme of \Cref{thm:main} with the integer $k$ and the $\MSO$ graph property $\phi = (F\operatorname{-minor-freeness})$ locally certifies the class of $F$-minor-free graphs with $O(\log n)$-bit labels.
\end{proof}




We emphasize again that the $O(\log n)$ upper bound of \Cref{thm:main} is \emph{optimal} because any proof labeling scheme that accepts paths and rejects cycles requires labels of size $\Omega(\log n)$~\cite{korman03032010}. Both paths and cycles are $\MSO$ properties and have bounded pathwidth (1 and 2, respectively).


\subsection{Additional related work}

Motivated by applications in self-stabilization, early works on local certification focused on checking the correctness of a given solution for a problem, such as verifying that a given set of edges is a spanning tree~\cite{korman03032010}. 
Other variations of the proof labeling schemes model exist. 
\citeauthor{fraginiaud2019random}~\cite{fraginiaud2019random} studied randomized proof labeling schemes.
\citeauthor{CENSORHILLEL2020112}~\cite{CENSORHILLEL2020112} introduced the approximate version of proof labeling schemes, where they studied proof labeling schemes to approximate graph diameter and maximum cardinality matching. \citeauthor{emek_et_al:LIPIcs:2020:13098}~\cite{emek_et_al:LIPIcs:2020:13098} further expanded the study of approximate proof labeling schemes to various optimization problems. 
\citeauthor{elek2022planarity}~\cite{elek2022planarity} and \citeauthor{esperet_et_al:LIPIcs.ICALP.2022.58}~\cite{esperet_et_al:LIPIcs.ICALP.2022.58} studied the property testing version of proof labeling schemes, where the soundness condition is relaxed: The verification algorithm is required to reject a problem instance only if it is $\epsilon$-far from satisfying the considered property.

\citeauthor{v012a019}~\cite{v012a019} introduced \emph{locally checkable proofs}, which are the proof labeling schemes that allow the verification algorithm to examine the labels within an $O(1)$-radius neighborhood and not just the immediate neighbors.  The study of proof labeling schemes was extended to \emph{distributed interactive proofs}, where the centralized prover and the distributed verifier can communicate with each other in multiple rounds~\cite{kol2018interactive,naor2020power}.

Local certification of various graph classes has attracted
significant attention in recent years: It was known that cographs~\cite{montealegre2021compact}, distance-hereditary graphs~\cite{montealegre2021compact}, and various geometric intersection graphs~\cite{jauregui2022distributed} admit compact distributed interactive proofs.

Beyond the realm of local certification, numerous studies have focused on designing distributed algorithms tailored to minor-closed graph classes by leveraging their structural properties.
Minor-closed graph classes were known to admit low-congestion shortcuts, which can be utilized to design fast distributed algorithms~\cite{ghaffari2021low,haeupler2016low,haeupler2016near,haeupler2018minor}.
\citeauthor{IzumiSPAA22}~\cite{IzumiSPAA22} presented an efficient distributed algorithm for computing tree decompositions, which was then used to obtain efficient algorithms for bounded-treewidth networks. \citeauthor{ghaffari2016planar}~\cite{ghaffari2016planar} developed efficient distributed algorithms for planarity testing, enabling several applications~\cite{ghaffari2016distributed,levi2021property}.
Minor-closed graph classes are also known to admit expander decompositions with good quality, leading to distributed algorithms that outperform those designed for general networks~\cite{chang2023efficient,chang2022narrowing}. Distributed approximation within minor-closed graph classes has been extensively studied, see the survey by Feuilloley~\cite{feuilloley2020bibliography}.

\section{Preliminaries}\label{sect:prelim}

In this section, we review the essential tools that we apply in our proofs.

\subsection{Basics tools in proof labeling schemes}
We consider a slight variation in the proof labeling scheme model.

\subparagraph{Edge certification.} Consider a variant of proof labeling schemes where the labels are placed on the edges and not on the vertices: The prover $\mathbf{P}$ assigns a label $\ell(e)$ to each edge $e \in E$, and the input of the verification algorithm $\mathbf{V}$ on vertex $v$ consists of its state $s(v)$ and the multiset $\{\ell(\{u,v\}) \, | \, u \in N(v)\}$ of the labels on the edges incident to $v$.

We say that a graph $G$ is \emph{$d$-degenerate} if the edges of $G$ can be oriented to form a directed acyclic graph with outdegree at most $d$. 

\begin{proposition}[\cite{feuilloley2023local}]\label{prop:degeneracy}
Suppose there is a proof labeling scheme for a graph class that is $d$-degenerate using $f(n)$-bit labels on edges.
The proof labeling scheme can be transformed into a proof labeling scheme using $O\left(d \cdot f(n)\right)$-bit labels on vertices.
\end{proposition}
\begin{proof}
Consider an edge orientation with outdegree at most $d$. To transform a given $f(n)$-bit edge certification into an $O(d \cdot f(n))$-bit vertex certification, it suffices that for each edge $e = u \rightarrow v$ oriented from $u$ to $v$, move the label $\ell(e)$ to $u$. 
\end{proof}

It is well-known~\cite{THOMASON2001318} that $H$-minor-free graphs are $d$-degenerate for  $d = O\left(t \sqrt{\log t}\right)$, where $t$ is the number of vertices in $H$. Therefore, as we focus on graph classes defined by forbidden minors, throughout the paper, we can freely put labels on edges in designing proof labeling schemes, as the overhead in the transformation of \Cref{prop:degeneracy} is only $O(1)$. 

\subparagraph{Graphs with designated vertices.}
We describe a proof labeling scheme to certify designated vertices with $O(\log n)$-bit edge labels.

\begin{proposition}
    \label{vertex-in-graph-edge-pls}
    Given a fixed $O(\log n)$-bit identifier $x$, there exists a proof labeling scheme with $O(\log n)$-bit edge labels to certify the existence of a vertex $v$ with $\ID(v) = x$.  
\end{proposition}

\begin{proof}
The proof, which uses the spanning tree technique~\cite{afek1991memory}, is folklore. The prover $\mathbf{P}$ selects a spanning tree $T$ rooted at the vertex $v$ with $\ID(v) = x$. For each edge $\{u,w\}$, its $O(\log n)$-bit label consists of $x$ and a distance label $\min(\dist(v,u), \dist(v,w))$, where $d(v,y)$ is the distance between $v$ and $y$ in $T$.

For the verification algorithm $\mathbf{V}$ at the vertex $v$ with $\ID(v) = x$, we check if $v$ is indeed the root of the tree $T$ by checking whether the distance label for all its incident edges is zero. For each remaining vertex $u$, we check if $u$ has exactly one parent in $T$ by examining whether there is an integer $d$ such that exactly one incident edge has distance label $d$ and all the remaining incident edges have distance label $d+1$. Finally, the requirement that all edges are marked with the same identifier $x$ can also be verified locally.

For correctness, observe that all vertices output $\acc$ if and only if the labeling encodes a spanning tree $T$ rooted at a vertex $v$ with $\ID(v) = x$. If there exists a vertex $v$ with $\ID(v) = x$, then all vertices output $\acc$ on the labeling created by $\mathbf{P}$. Otherwise, regardless of the given labeling, at least one vertex outputs $\rej$, as such a spanning tree $T$ does not exist.
\end{proof}


In the subsequent discussion, we informally refer to the proof labeling scheme in \Cref{vertex-in-graph-edge-pls} as the proof labeling scheme \emph{pointing to} $v$. 

\subsection{\texorpdfstring{$\MSO$}{MSO2} properties and \texorpdfstring{$k$}{k}-terminal recursive graphs}

The proof of Courcelle's theorem, which states that any $\MSO$ property is decidable in linear time for graphs of bounded treewidth~\cite{borie1992automatic,COURCELLE199012}, relies on decomposing bounded-treewidth graphs into $k$-terminal recursive graphs, defined as follows.

\begin{definition}
    \label{terminal-recursive-gluing}
    A \underline{$k$-terminal graph} is a graph $G$ equipped with an ordered subset $T_G \subseteq V[G]$ of at most $k$ vertices called \ul{terminals}. Given two functions $f_1, f_2 : [k] \to [k] \cup \{\perp\}$, the \ul{composition operator} $\odot_{f_1, f_2}$, when given two $k$-terminal graphs $G_1, G_2$ as input, returns the union of $G_1$ and $G_2$, whose $i$th terminal is the $f_1(i)$th terminal of $G_1$ and the $f_2(i)$th terminal of $G_2$; if both of them not $\perp$, then the two terminals are glued together.
 Let $B$ denote the set of all $k$-terminal graphs with at most $k$ vertices. The set of \ul{$k$-terminal recursive graphs} is defined as the closure of $B$ under all possible composition operations.
\end{definition}

\begin{figure}[ht!]
    \centering
    \includegraphics[scale=0.65]{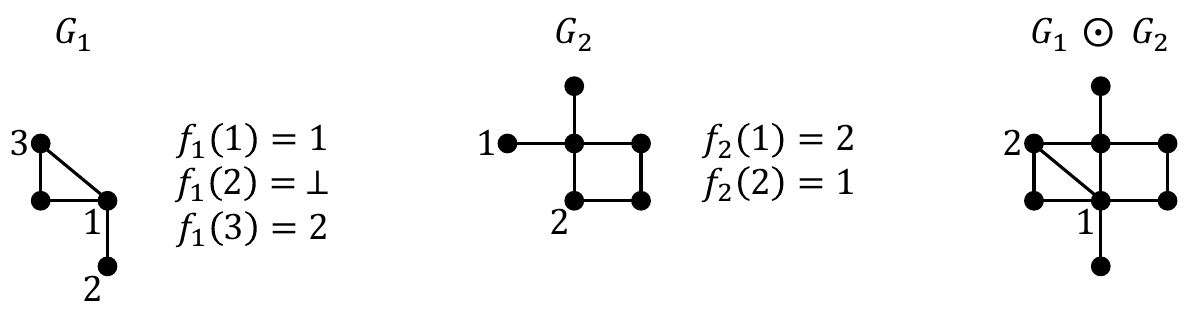}
    \caption{Combining two $3$-terminal graphs into a $3$-terminal graph.}
    \label{fig:f2}
\end{figure}

See \Cref{fig:f2} for an illustration of \Cref{terminal-recursive-gluing}. 
The definition of \emph{$k$-terminal recursive graphs} is usually more general in the literature, allowing for composition operations with arity greater than two.
For our purposes, the above definition suffices.
\begin{proposition}[{\cite{borie1992automatic,COURCELLE199012}}]
    \label{homomorphism-class-mapping}
    Let $F_k$ be the set of $k$-terminal recursive graphs, and $\phi$ be an $\MSO$ property, then there exists a function $h : F_k \to C$, for some finite set $C$ of \ul{homomorphism classes} satisfying the following conditions.
    \begin{itemize}
        \item Given any two graphs $G_1, G_2 \in F_k$ such that $h(G_1) = h(G_2)$, $G_1$ satisfies $\phi$ if and only if $G_2$ satisfies $\phi$.
        \item For any composition operator $\odot$, there exists a function $f_{\odot} : C^2 \to C$ such that 
        \[ f_{\odot}(h(G_1), h(G_2)) = h(G_1 \odot G_2)\]
    \end{itemize}
\end{proposition}

Courcelle’s theorem is then obtained by interpreting a bounded-treewidth graph as a $K$-terminal recursive graph, and using \Cref{homomorphism-class-mapping} to calculate its homomorphism class recursively by dynamic programming.

As previously discussed, Courcelle’s theorem and \Cref{homomorphism-class-mapping} are applicable to graphs with a specified vertex subset $X \subseteq V$ or an edge set $Y \subseteq E$, or more generally, to graphs where vertices and edges are labeled with inputs from a fixed finite set. This allows us to capture properties such as ``$X$ is a dominating set of $G$'' and ``$Y$ induces a planar subgraph of $G$,'' where $X \subseteq V$ is represented as an one-bit vertex label and $Y \subseteq E$ is represented as an one-bit edge label.




\section{Technical overview}\label{sect:overview}

We begin with a review of the previous $O(\log^2 n)$-bit proof labeling schemes for $\MSO$ properties in bounded-treewidth graphs by Fraigniaud, Montealegre, Rapaport, and Todinca~\cite{fraigniaud2024meta}. As the class of graphs of treewidth at most $k$ is the same as the class of $(k+1)$-terminal recursive graphs~\cite{bodlaender1998partial}, their idea is to certify an execution of the dynamic programming algorithm for Courcelle's theorem by storing relevant information at each bag of the tree decomposition. The main technical challenge to realize this approach is that each bag of the tree decomposition can be \emph{disconnected}, requiring information to be routed between vertices in the same bag and introducing congestion overhead.

To address this issue, they utilized the following well-known result: At the cost of increasing the width by a factor of three, the depth of the tree decomposition can be made $O(\log n)$~\cite{bodlaender1989nc}. Furthermore, by examining additional structural properties of this decomposition, they showed that the required routing can be performed with $O(\log n)$ congestion. This congestion overhead explains why the resulting label size is $O(\log^2 n)$ rather than the desired $O(\log n)$. Since any tree decomposition of bounded width necessarily has depth $\Omega(\log n)$, breaking the barrier $O(\log^2 n)$ requires a different approach.

\subparagraph{New approach.} In this work, we demonstrate that for bounded-pathwidth graphs, the optimal label size of $O(\log n)$ is attainable. Instead of relying on traditional path and tree decompositions, we develop a new graph decomposition that allows us to reduce the congestion overhead of executing the dynamic programming algorithm for Courcelle’s theorem to $O(1)$.

We begin with the interval representation corresponding to the path decomposition. We show that the intervals can be arranged into $O(1)$ lanes of pairwise-disjoint intervals in such a way that additional edges can be  embedded into the graph with $O(1)$ congestion to complete the lanes as follows:
\begin{itemize}
    \item The vertices within each lane are connected to form a path.
    \item The initial vertices of all lanes are connected to form a path.
\end{itemize}

Similar to the relation between treewidth and the class of $k$-terminal recursive graphs~\cite{bodlaender1998partial}, we prove that such graphs can be recursively constructed using the following two merging operations on $k$-lane graphs, which are graphs $G$ with a set of lanes $\mathcal{T}(G) \subseteq [k]$, where each lane $i \in \mathcal{T}(G)$ is assigned an in-terminal $\tin_i(G)$ and an out-terminal $\tout_i(G)$. 

 In a $\bmerge$ operation, two $k$-lane graphs over disjoint sets of lanes are combined by adding an edge between an out-terminal of one graph and an out-terminal of the other graph. 
 
 In a $\tmerge$ operations, we are given a tree $T$ over $k$-lane graphs satisfying the following conditions:
 \begin{itemize}
    \item If $G_1$ is a child of $G_2$ in $T$, then $\mathcal{T}(G_1) \subseteq \mathcal{T}(G_2)$.
    \item If $G_1$ and $G_2$ have the same parent in $T$, then $\mathcal{T}(G_1) \cap \mathcal{T}(G_2) = \emptyset$.  
\end{itemize}
The $k$-lane graphs in $T$ are then merged as follows: For each graph $G$ in $T$ and each of its child $G'$, identifying each in-terminal $\tin_i(G')$ of $G'$ with the out-terminal $\tout_i(G)$ of $G$ in the same lane. 

A $k$-lane recursive graph is a $k$-lane graph that can be recursively constructed by these operations. We prove that the class of $k$-lane recursive graphs has the following favorable properties: 
\begin{itemize}
    \item Any $k$-lane recursive graph can be constructed with a \emph{bounded depth} of recursion. 
    \item Any $k$-lane recursive graph is \emph{connected}.
\end{itemize}
Since a $O(1)$-lane recursive graph can be viewed as a $O(1)$-terminal recursive graph, the aforementioned desirable properties enable us to certify the execution of the dynamic programming algorithm for Courcelle's theorem with optimal $O(\log n)$-bit proof labeling schemes.

\subparagraph{Roadmap.}
The rest of the paper proceeds as follows. 
In \Cref{sect:lane-decomposition}, we show that, for any interval representation associated with a bounded-pathwidth graph, we can partition the vertices into $O(1)$ lanes such that the lane partition can be completed with $O(1)$ congestion.
In \Cref{sect:construction}, we show that the resulting graphs are $k$-lane recursive graphs with $k \in O(1)$.
In \Cref{sect:certification}, we present a proof labeling scheme to certify any given $\MSO$ graph property using $O(\log n)$-bit labels on $k$-lane recursive graphs.
In \Cref{sect:conclusion}, we conclude the paper.



\section{Completing a \texorpdfstring{$k$}{k}-lane partition with low congestion}\label{sect:lane-decomposition}

Consider an interval representation of a graph, where the intervals are partitioned into lanes.
Our goal is to augment the graph by adding edges such that each lane forms a path, and the initial vertices of all lanes are concatenated into a single path. Each newly added edge is embedded as a path in the original graph. 
We demonstrate that if the width of the given interval representation is $O(1)$, then there exists an $O(1)$-lane partition that allows the desired embedding with $O(1)$ congestion.
We present the necessary definitions in \Cref{sect:interval} and prove the main result in \Cref{sect:embedding}.

\subsection{Graph terminology}\label{sect:interval}

We start with a formal definition of an interval decomposition. 

\begin{definition}
An \ul{interval representation} $\mathcal{I} = \{I_v \, | \, v \in V\}$ of a graph $G=(V,E)$ is an assignment of a non-empty interval $I_v = [L_v, R_v]$ to each vertex $v \in V$ such that $I_u \cap I_v \neq \emptyset$ for each edge $e = \{u,v\}\in E$. The \ul{width} of an interval decomposition is the maximum number of intervals with a non-empty intersection.
\end{definition}


By \Cref{pathwidth-defn}, we know that a graph has pathwidth $k$ if and only if it has an interval representation of width $k+1$.
In the subsequent discussion, for any two non-empty intervals $[a,b]$ and $[c,d]$, we write $[a,b] \prec [c,d]$ if $b < c$, that is, $[a,b]$ is strictly before $[c,d]$.

\begin{definition}
A \ul{$k$-lane partition}  $\mathcal{P}=(P_1, \ldots P_k)$  of an interval representation $\mathcal{I}$ of a graph $G=(V,E)$ is a partition of the vertex set $V$ into $k$ non-empty sequences $P_1, \ldots, P_k$ such that for each $i \in [k]$, the sequence $P_i=\left(v_1^{i}, \ldots, v_{|P_i|}^i\right)$ satisfies $I_{v_1^{i}} \prec \cdots \prec I_{v_{|P_i|}^{i}}$.
\end{definition}

In other words, a $k$-lane partition is a partition of a collection of intervals into $k$ lanes of pairwise-disjoint intervals, ordered according to $\prec$. The following observation implies that any interval representation of width $k$ admits a $k$-lane partition.

\begin{observation}
    \label{greedy-interval-split}
Any collection $\mathcal{I}$ of non-empty intervals of width $k$ can be partitioned into $k$ subsets $\mathcal{I}=\mathcal{I}_1 \cup \cdots \cup \mathcal{I}_k$ such that for each $i \in [k]$, the intervals in $\mathcal{I}_i$ are pairwise-disjoint.
\end{observation} 
\begin{proof}
The observation is a restatement of the well-known fact that the size of a maximum clique equals the chromatic number for interval graphs. 
\end{proof}

\begin{definition}\label{def:completion}
Given a {$k$-lane partition}  $\mathcal{P}=(P_1, \ldots P_k)$  of an interval representation $\mathcal{I}$ of a graph $G=(V,E)$, let
\[E_1 = \left\{ \{v_j^i, v_{j+1}^{i}\} \, | \, 1 \leq i \leq k \text{ and } \ 1 \leq j < |P_i| \right\} \ \ \text{ and } \  \  E_2= \left\{ \{v_1^i, v_{1}^{i+1}\} \, | \, 1 \leq i < k \right\}.\]
Define the \ul{completion} and \ul{weak completion} of $(G,\mathcal{I},\mathcal{P})$ as $(V,E \cup E_1 \cup E_2)$ and $(V,E \cup E_1)$, respectively.
\end{definition}

See \Cref{fig:f3} for an illustration of \Cref{def:completion}. In other words, a weak completion is obtained by adding new edges to transform each lane into a path. If we further connect the initial vertices of all lanes to form a path, the result is called a completion. As we will later see, these additional edges $E_1 \cup E_2$ are useful in designing proof labeling schemes.
 
\begin{figure}[ht!]
    \centering
    \includegraphics[scale=0.65]{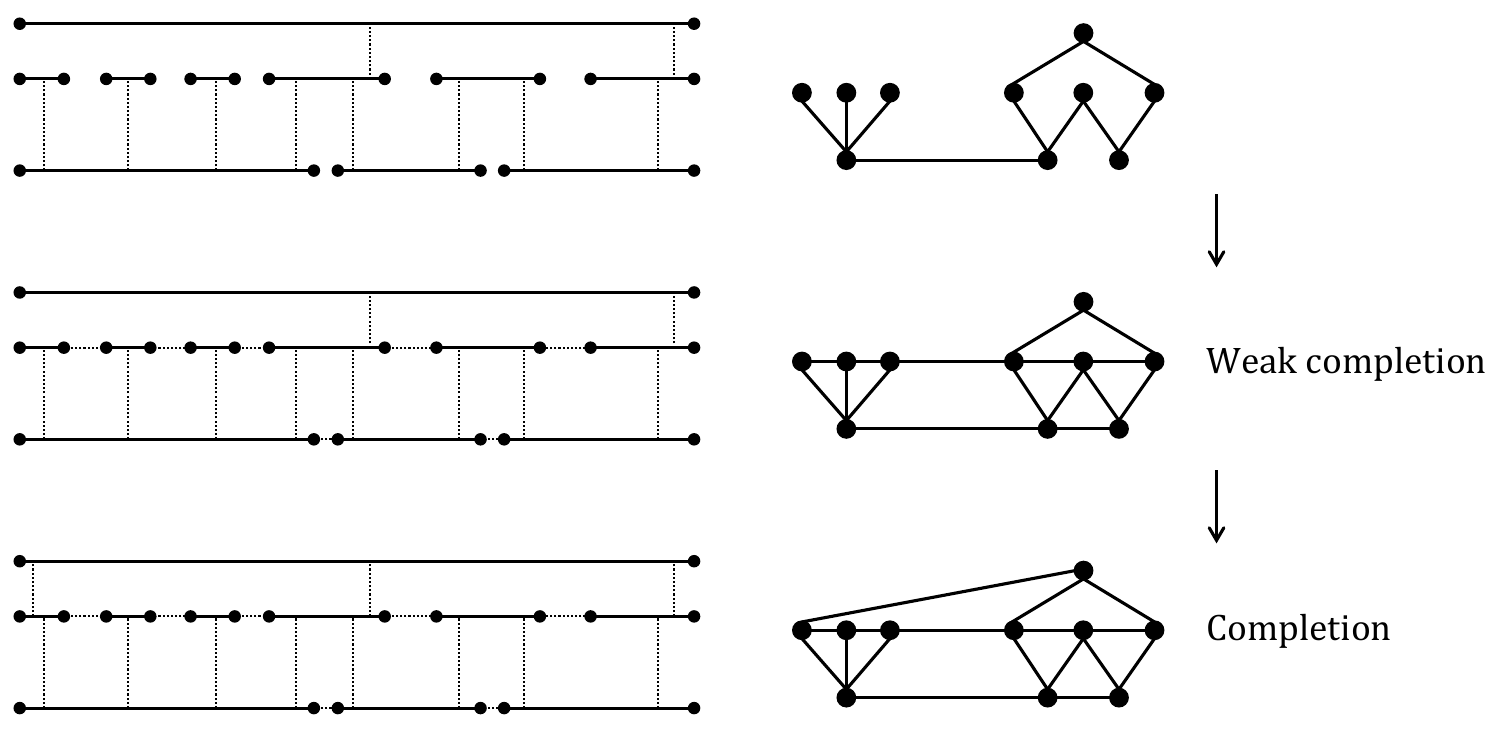}
    \caption{Weak completion and completion.}
    \label{fig:f3}
\end{figure}

\begin{definition}
Let $G=(V,E)$ and $G'=(V, E')$ be two graphs on the same vertex set $V$ with $E \subseteq E'$. An \ul{embedding} $\mathcal{E}=\{ P_e \, | \, e\in E' \setminus E\}$ of $G'$ into $G$ is an assignment of a $u$-$v$ path $P_e$ in $G$ to each edge $e=\{u,v\} \in E' \setminus E$. The \ul{congestion} of an embedding $\mathcal{E}$ is the maximum number of paths in $\mathcal{E}$ that an edge belongs to.
\end{definition}

We explain the motivation behind the above definition, as follows.
Suppose we are given a proof labeling scheme on $G'=(V, E')$ with $b$-bit edge labels.
If we can find an embedding of $G'$ into $G$ with congestion $c$, then the $b$-bit edge labels in $E' \setminus E$ can be simulated in $G=(V,E)$ at a cost of $O(bc)$ bits by putting the label of each $e\in E' \setminus E$ on all edges in $P_e$.

\subsection{Low-congestion embedding}\label{sect:embedding}
The main goal of this section is to show that for any given $O(1)$-width interval representation $\mathcal{I}$ of a graph $G$, there exists an $O(1)$-lane partition $\mathcal{P}$ such that the completion of $(G,\mathcal{I},\mathcal{P})$ can be embedded into $G$ with $O(1)$ congestion. To do so, we define the following three functions.
\begin{align*}
   f(k) &=\begin{cases}
    1, &\text{if $k=1$,}\\
    2 + 2(k-1) \cdot f(k-1), &\text{if $k > 1$.}
\end{cases}\\
   g(k) &=\begin{cases}
    0, &\text{if $k=1$,}\\
    2 + g(k-1) + 2k \cdot f(k-1), &\text{if $k > 1$.}
\end{cases}\\
h(k) &= g(k)+f(k)-1
\end{align*}

\begin{proposition}
\label{low-completion-cost}
For any given interval representation $\mathcal{I}$ of a connected graph $G=(V,E)$ with width $k$, there exists a $w$-lane partition $\mathcal{P}$ with $w \leq f(k)$ such that the following holds.
\begin{itemize}
    \item The weak completion of $(G,\mathcal{I},\mathcal{P})$ can be embedded into $G$ with congestion at most $g(k)$.
    \item The completion of $(G,\mathcal{I},\mathcal{P})$ can be embedded into $G$ with congestion at most $h(k)$.
\end{itemize}
\end{proposition}

The second statement of \Cref{low-completion-cost} follows immediately from the first statement of \Cref{low-completion-cost}: By selecting an arbitrary $u$-$v$ path $P_e$ for each $e \in E_2$, the additional cost of embedding $E_2$ is at most $|E_2| = w-1 \leq f(k)-1$. Therefore, in the subsequent discussion, we focus on the first statement, which is proved by an induction on $k$.

\subparagraph{Base case.} For the base case of $k = 1$, $G$ must be an isolated vertex. By selecting the trivial partition $\mathcal{P}$ with $f(1) = 1$ lane, the weak completion of $(G,\mathcal{I},\mathcal{P})$ can be embedded into $G$ with congestion $g(1) = 0$, as $E_1 = \emptyset$ and $E_2 = \emptyset$.

\subparagraph{Inductive step.} Now, assuming the claim is true for $1, 2, \ldots, k-1$, we prove that it is also true for $k$. To do so, we describe the construction of the needed $w$-lane partition $\mathcal{P}$ and the required embedding.  We start with some definitions. See \Cref{fig:f4} for an illustration of these definitions.

\begin{figure}[ht!]
    \centering
    \includegraphics[scale=0.65]{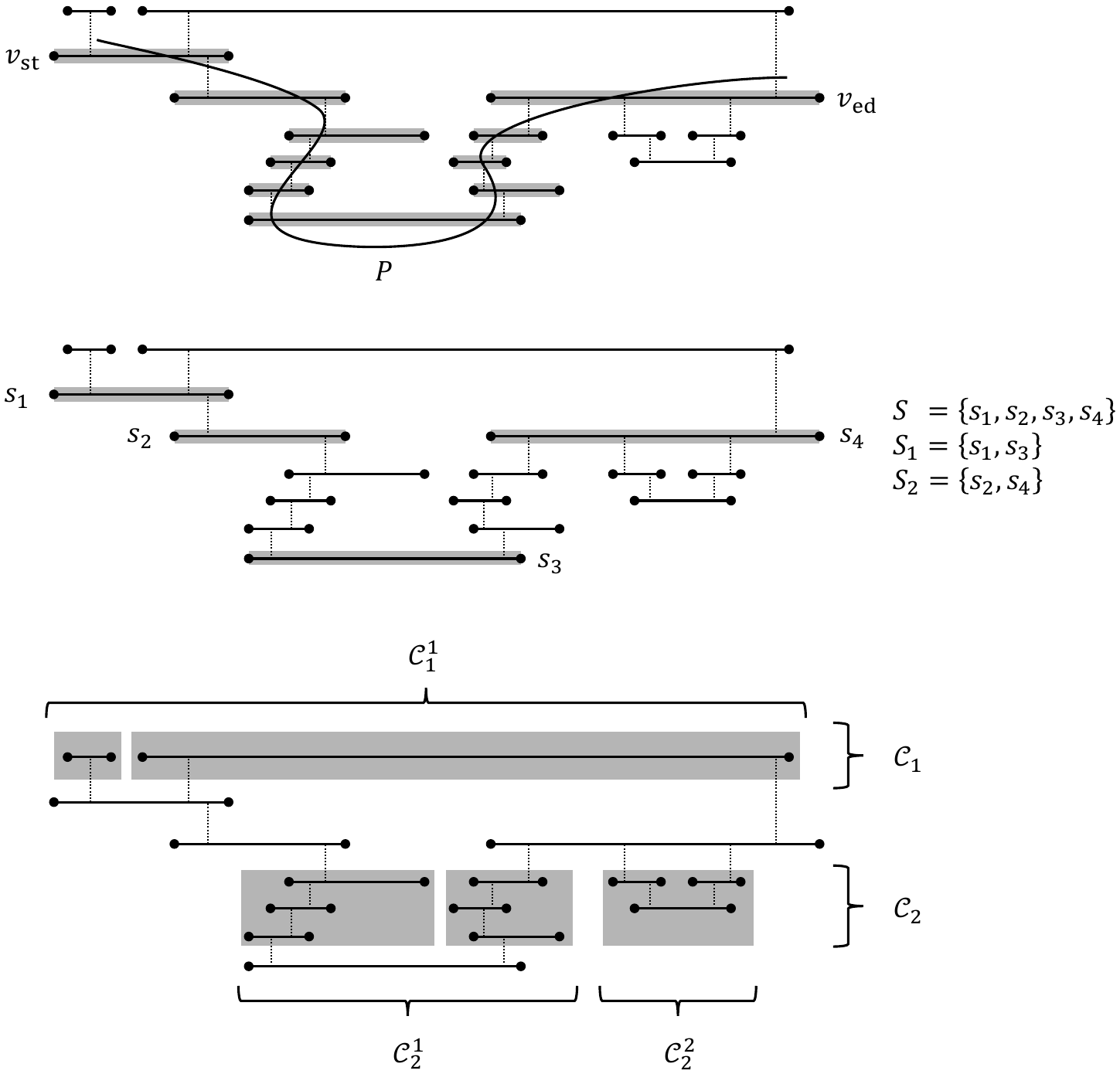}
    \caption{Graph terminology in \Cref{sect:embedding}.}
    \label{fig:f4}
\end{figure}

\begin{itemize}
    \item Let $\vst \in V$ be chosen to minimize  $L_{\vst}$.
\item Let $\ved \in V$ be chosen to maximize $R_{\ved}$.
\item Let $P$ be any $\vst$-$\ved$ path in $G$. 
\item Let $S=\{s_1, s_2, \ldots\}$ be a sequence of vertices defined as follows.
\begin{itemize}
    \item Select $s_1 = \vst$ to be the first vertex in $P$. 
    \item If $R_{s_i} < R_{\ved}$, then select $s_{i+1}$ to maximize $R_{s_{i+1}}$ from the set of vertices $u$ in $P$ \emph{after} $s_i$ such that $I_u \cap I_{s_{i}} \neq \emptyset$. Such a vertex $s_{i+1}$ exists, since otherwise $P$ is disconnected.
\end{itemize}
\item Let $S_1 = \{s_1, s_3, s_5, \ldots\}$ and $S_2 = \{s_2, s_4, s_6, \ldots\}$.
\end{itemize}

We make the following observations.

\begin{observation}\label{obs0}
For all $1 \leq i < |S|$, $R_{s_i} < R_{s_i+1}$.
\end{observation}
\begin{proof}
If $R_{s_i} \geq R_{s_i+1}$, then $R_{u} \leq R_{s_i}$ for all vertices $u$ in $P$ {after} $s_i$ such that $I_u \cap I_{s_{i}} \neq \emptyset$, implying that $R_{s_i} = R_{\ved}$, so the construction of $S$ stops at $v_i$, contradicting the existence of $v_{i+1}$.
\end{proof}

\begin{observation}\label{obs1}
$I_{s_1} \prec I_{s_3} \prec I_{s_5} \prec \cdots$ and $I_{s_2} \prec I_{s_4} \prec I_{s_6} \prec \cdots$.
\end{observation}
\begin{proof}
To prove the observation, it suffices to show that $I_{s_i} \prec I_{s_{i+2}}$ for every $1 \leq i \leq |S|-2$. Suppose the claim is false, meaning that there exists an index $i$ such that $L_{s_{i+2}} \leq R_{s_i}$. By \Cref{obs0}, $R_{s_i} < R_{s_{i+2}}$, so $I_{s_{i+2}}=[L_{s_{i+2}}, R_{s_{i+2}}]$ and $I_{s_{i}}=[L_{s_{i}}, R_{s_{i}}]$ overlaps. By \Cref{obs0} again, $R_{s_{i+1}} < R_{s_{i+2}}$, contradicting the choice of $v_{i+1}$.
\end{proof}

For any set of vertices $U$, we write $I_U= \bigcup_{v \in U} I_v$. Observe that if $U$ is \emph{connected}, then $I_U=[L_U,R_U]$, where $L_U = \min_{v \in U} L_v$ and $R_U = \max_{v \in U} R_v$.  

 \begin{observation}\label{obs2}
$I_S=[L_S,R_S]$ satisfies $L_S = \min_{v \in V} L_v$ and $R_S = \max_{v \in V} R_v$. 
\end{observation}
\begin{proof}
The observation follows from the following facts.
\begin{itemize}
    \item $\min_{v \in S} L_v = L_{\vst} = \min_{v \in V} L_v$. 
    \item $\max_{v \in S} R_v = R_{\ved} = \min_{v \in V} L_v$. 
    \item $I_{s_i} \cap I_{s_{i+1}} \neq \emptyset$ for every $1 \leq i < |S|$. \qedhere
\end{itemize}
\end{proof}

Let $G'$ be the subgraph of $G$ induced by $V \setminus S$. Let $\mathcal{C}$ be the set of connected components of $G'$. We partition $\mathcal{C}$ into $k$
subsets  $\mathcal{C}_1, \ldots, \mathcal{C}_{k-1}$ according to the following lemma.

\begin{lemma}\label{lem:disjoint}
There exists a partition of $\mathcal{C} = \mathcal{C}_1 \cup  \cdots \cup 
 \mathcal{C}_{k-1}$ such that for every $i \in [k-1]$, the set of intervals
$\{ I_C \, | \,  C \in \mathcal{C}_i \}$
is pairwise-disjoint.
\end{lemma}

\begin{proof}
    In view of \Cref{greedy-interval-split}, it suffices to show that the width of the collection of intervals $\{ I_C \, | \,  C \in \mathcal{C}_i \}$ is at most $k-1$. If the width of  $\{ I_C \, | \,  C \in \mathcal{C}_i \}$  is at least $k$, then there exists a point $x \in [\min_{v \in V} L_v, \max_{v \in V} R_v]$ that belongs to $k$ intervals in $\{ I_C \, | \,  C \in \mathcal{C}_i \}$, so we can find $k$ vertices in $V \setminus S$ whose interval contains the point $x$. By \Cref{obs2}, the interval of some vertex in $S$ contains $x$. In total, we find $k+1$ intervals in $\mathcal{I}$ overlapping at $x$, contradicting the assumption that the width of $\mathcal{I}$ is $k$. Thus, the lemma holds. 
\end{proof}


We further refine the partition $\mathcal{C} = \mathcal{C}_1 \cup  \cdots \cup 
 \mathcal{C}_{k-1}$ by partitioning each $\mathcal{C}_i$ into $\mathcal{C}_i = \mathcal{C}_i^1 \cup \mathcal{C}_i^2$, as follows. Consider a connected component $C \in \mathcal{C}_i$.  If there exist a vertex $u$ in $C \in \mathcal{C}_i$ and a vertex $v$ in $S_1$ such that $\{u,v\}\in E$, then $C \in \mathcal{C}_i^1$. Otherwise,  $C \in \mathcal{C}_i^2$. Since $G$ is connected, if $C \in \mathcal{C}_i^2$, then there exist a vertex $u$ in $C \in \mathcal{C}_i$ and a vertex $v$ in $S_2$ such that $\{u,v\}\in E$. 


For each $C \in \mathcal{C}$, we write $\mathcal{I}_C$ to denote the interval representation $\mathcal{I}$ restricted to $C$ and write $G_C$ to denote the subgraph of $G$ induced by $C$.

\begin{lemma}
    \label{induct-partition-component}
    For all $C \in \mathcal{C}$, the width of the interval representation $\mathcal{I}_C$ is at most $k-1$.
\end{lemma}
\begin{proof}
Since the width of $\mathcal{I}$ is $k$, any point $x \in [\min_{v \in V} L_v, \max_{v \in V} R_v]$  belongs to at most $k$ intervals in $\mathcal{I}$.
By \Cref{obs2}, the interval of some vertex in $S$
already contains $x$. Therefore, $x$ belongs to at most $k-1$ intervals in $\mathcal{I}_C$.  
\end{proof}

 By \Cref{induct-partition-component} and the inductive hypothesis, for all $C \in \mathcal{C}$, there exists a $w$-lane partition $\mathcal{P}_C$ of the interval representation $\mathcal{I}_C$ with $w \leq f(k-1)$ such that the weak completion of $(G_C,\mathcal{I}_C,\mathcal{P}_C)$ can be embedded into $G_C$ with congestion at most $g(k-1)$. For each $\ell \in [f(k-1)]$, we write ${P}_C^\ell$ to denote the $\ell$th lane of $\mathcal{P}_C$. In case $\ell > w$, $\mathcal{P}_C^\ell = \emptyset$.

 \subparagraph{Partitioning into lanes.} We are now ready to describe the $w$-lane partition $\mathcal{P}$ of $\mathcal{I}$ for the proof of \Cref{low-completion-cost}.
 \begin{itemize}
     \item For both $i \in \{1,2\}$, $S_i$ is one lane.
     \item For all $i \in [k-1]$, $j \in \{1,2\}$, and $\ell \in [f(k-1)]$, $\bigcup_{C \in \mathcal{C}_i^j} {P}_C^\ell$ is one lane.
 \end{itemize}
 There are $f(k) = 2 + 2(k-1)f(k-1)$ lanes in total, some of which might be empty, so the total number $w$ of non-empty lanes is at most $f(k)$, as required.

 For the validity of the construction, we verify that each lane consists of pairwise-disjoint intervals so that they can be sequentially ordered according to $\prec$. This is true for the two lanes $S_1$ and $S_2$ because of \Cref{obs1}. For $\bigcup_{C \in \mathcal{C}_i^j} {P}_C^\ell$, by the induction hypothesis, ${P}_C^\ell$ consists of pairwise-disjoint intervals within the range $I_C$, and from \Cref{lem:disjoint} we know that the set of intervals $\{ I_C \, | \,  C \in \mathcal{C}_i \}$ is pairwise-disjoint.

 \begin{proof}[Proof of \Cref{low-completion-cost}]
To finish the proof of \Cref{low-completion-cost}, it suffices to show that the weak completion of $(G,\mathcal{I},\mathcal{P})$ can be embedded into $G$ with congestion at most $g(k)$. We divide the edge set $E_1 = \left\{ \{v_j^i, v_{j+1}^{i}\} \, | \, 1 \leq i \leq w \text{ and } \ 1 \leq j < |P_i| \right\}$ to be embedded into cases and analyze their congestion cost separately.

\subparagraph{Case 1.} In the first case, we consider the cost of transforming the two lanes $S_1$ and $S_2$ into paths. To do so, for each $e=\{s_{i}, s_{i+2}\}$, we set $P_e$ to be the subpath from $s_i$ to $s_{i+2}$ in $P$. Since each edge in $P$ is used at most twice, the overall congestion cost is at most $2$. See \Cref{fig:f5} for an illustration.
\begin{figure}[ht!]
    \centering
    \includegraphics[scale=0.65]{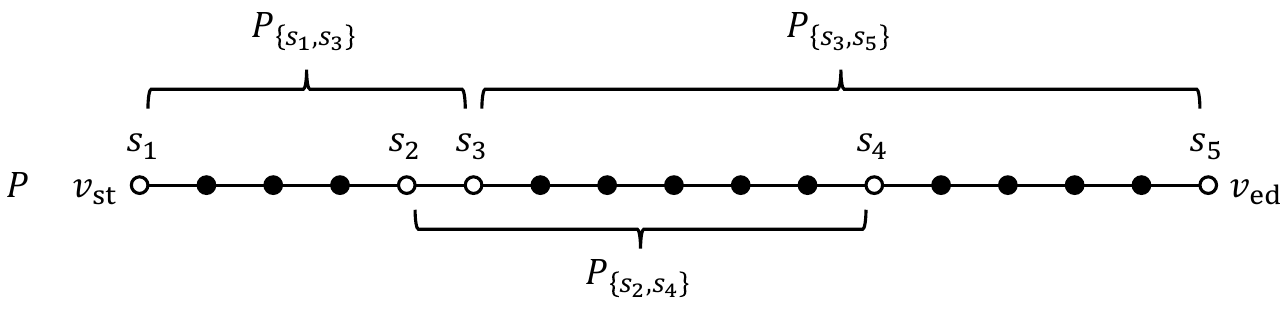}
    \caption{Case 1 in the proof of \Cref{low-completion-cost}.}
    \label{fig:f5}
\end{figure}

\subparagraph{Case 2.} In the second case, we consider the cost of transforming the lane $\bigcup_{C \in \mathcal{C}_i^j} {P}_C^\ell$ into a path for all $i \in [k-1]$, $j \in \{1,2\}$, and $\ell \in [f(k-1)]$. We further divide the analysis into two subcases. 

\subparagraph{Case 2.1.} We consider the cost of embedding the edges within the same connected component $C \in \mathcal{C}_i^j$. Recall that the weak completion of $(G_C,\mathcal{I}_C,\mathcal{P}_C)$ can be embedded into $G_C$ with congestion at most $g(k-1)$ by \Cref{induct-partition-component} and the inductive hypothesis. Since $\{G_C \, | \, C \in \mathcal{C} \}$ are disjoint subgraphs, applying the embedding for all $C \in \mathcal{C}$, the overall congestion cost is still at most $g(k-1)$.

\subparagraph{Case 2.2.} We consider the cost of embedding the edges crossing two connected components $C$ and $C'$ in $\mathcal{C}_i^j$. Let us focus on one such edge $e = \{x,y\}$, where $x \in C$ is the last vertex in the lane $\mathcal{P}_{C}^\ell$ and $y \in C'$ is the first vertex in the lane $\mathcal{P}_{C'}^\ell$. 

Since $C \in \mathcal{C}_i^j$, there exists an edge $e_C^\ast$ in $G$ between a vertex $u_C^\ast \in C$ and a vertex $v_C^\ast \in S_j$. Similarly, since $C' \in \mathcal{C}_i^j$, there exists an edge $e_{C'}^\ast$ in $G$ between a vertex $u_{C'}^\ast \in C'$ and a vertex $v_{C'}^\ast \in S_j$. Using these two edges, we select the $x$-$y$ path \[P_{e}=(x, \ldots, u_C^\ast, v_C^\ast, \ldots, v_{C'}^\ast, u_{C'}^\ast, \ldots, y),\] as follows. See \Cref{fig:f6} for an illustration.
\begin{itemize}
    \item The subpath $(x, \ldots, u_C^\ast)$ is chosen as any path from $x$ to $u_C^\ast$ in $C$.
    \item The subpath $(v_C^\ast, \ldots, v_{C'}^\ast)$ is chosen as the unique path from $v_C^\ast$ to $v_{C'}^\ast$ in $P$.
    \item The subpath $(u_{C'}^\ast, \ldots, y)$ is chosen as path from $u_{C'}^\ast$ to $y$ in $C'$.
\end{itemize}
We now analyze the overall congestion cost to perform the embedding for all edges crossing two connected components. Recall that for each $C \in \mathcal{C}$, $\mathcal{P}_C$ is a partition into at most $f(k-1)$ lanes, and $\mathcal{C}$ consists of pairwise-disjoint vertex sets, so the overall congestion cost to embed the prefix $(x, \ldots, u_C^\ast, v_C^\ast)$ is at most $f(k-1)$. Similarly, the congestion cost for the suffix $(v_{C'}^\ast, u_{C'}^\ast, \ldots, y)$ is also at most $f(k-1)$.

\begin{figure}[ht!]
    \centering
    \includegraphics[scale=0.65]{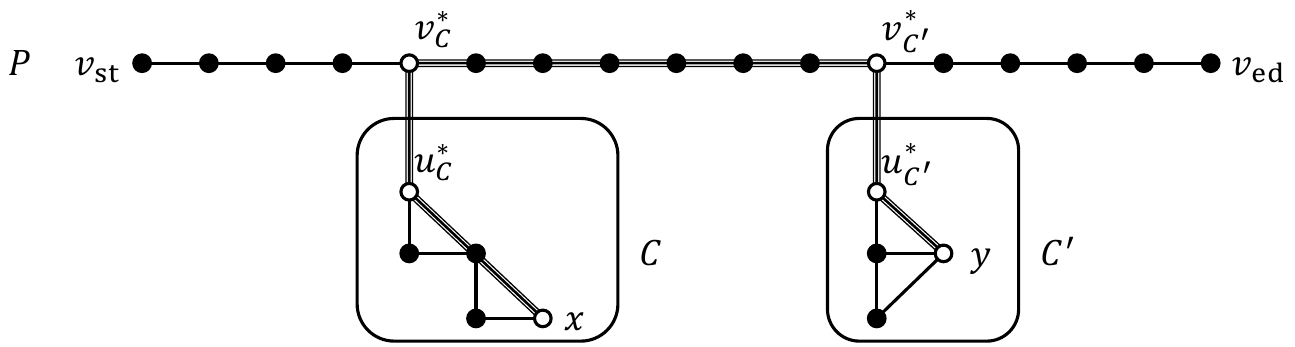}
    \caption{Case 2.2 in the proof of  \Cref{low-completion-cost}.}
    \label{fig:f6}
\end{figure}

Next, we consider the middle portion $(v_C^\ast, \ldots, v_{C'}^\ast)$ of the path $P_{e}$. For each $(i,j,\ell)$, each edge in $P$ participates at most once in the
transformation of the lane $\bigcup_{C \in \mathcal{C}_i^j} {P}_C^\ell$ into a path. Therefore, the overall cost is at most the number of choices of $(i,j,\ell)$, which is $2 \cdot (k-1) \cdot f(k-1)$.
To summarize, the overall cost for Case 2.2 is upper bounded by  $f(k-1) + f(k-1) + 2 \cdot (k-1) \cdot f(k-1) = 2k\cdot f(k-1)$.

\subparagraph{Summary.} Considering all cases, the congestion of the embedding is upper bounded by $2 + g(k-1) + 2k\cdot f(k-1) = g(k)$, as required. Thus, we have proved the first statement of the theorem. As discussed before, the second statement follows immediately from the first statement, as we just need to embed at most $f(k)-1$ new paths to turn the weak completion into the completion.
\end{proof}

\section{Constructing \texorpdfstring{$k$}{k}-lane recursive graphs}\label{sect:construction}

We define the \emph{lanewidth} of a graph using the following two operations on graphs with $k$ distinct designated vertices $\{\tau_1, \ldots, \tau_k\}$. See \Cref{fig:f7} for an example.

\begin{figure}[ht!]
    \centering
    \includegraphics[scale=0.65]{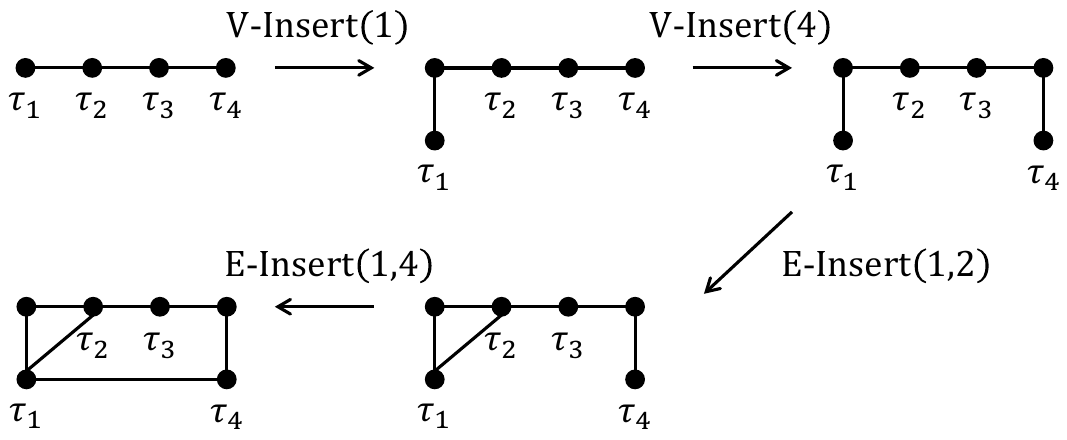}
    \caption{A bounded-lanewidth graph.}
    \label{fig:f7}
\end{figure}

\begin{description}
    \item[$\insv(i)$:] Add a vertex $v$ with an edge $\{v, \tau_i\}$, and reset the $i$th designated vertex to $v$.
    \item[$\inse(i,j)$:] Add an edge $\{\tau_i, \tau_j\}$.
\end{description}

\begin{definition}
\label{lanewidth-defn}
The \ul{lanewidth} of a graph $G$ is the minimum value of $k$ such that $G$ can be constructed from a $k$-vertex path $P = (\tau_1, \ldots, \tau_k)$ using the two operations $\insv$ and $\inse$. 
\end{definition}

In \Cref{sect:lanewidth_def}, we give an alternative definition of lanewidth using the completion of a $k$-lane partition. 
In \Cref{sect:recursive_merge}, we introduce the notion of $k$-lane recursive graphs.
In \Cref{sect:recursive_merge,sect:recursive_def}, we consider a specific way to construct $k$-lane recursive graphs with a recursion of bounded depth.
In \Cref{sect:recursive_construct}, we show that any graph of lanewidth $k$ can be constructed in this manner.

\subsection{An alternative definition}\label{sect:lanewidth_def}
In this section, we show that the lanewidth of a graph can be defined alternatively in terms of the completion of a $k$-lane partition. 

\begin{proposition}
\label{lanewidth-lem}
The lanewidth of $G=(V,E)$ is the minimum value of $k$ such that there exist a graph $G'$, an interval representation $\mathcal{I}'$ of $G'$, and a $k$-lane partition $\mathcal{P}'$ of $\mathcal{I}'$ such that $G$ is the completion of $(G', \mathcal{I}', \mathcal{P}')$.
\end{proposition}

\begin{proof}
It suffices to prove that the following two statements are equivalent.
\begin{enumerate}
    \item \label{item1} $G$ can be constructed from $P = (\tau_1, \ldots, \tau_k)$ using the two operations $\insv$ and $\inse$. 
    \item\label{item2} There exist a graph $G'$, an interval representation $\mathcal{I}'$ of $G'$, and a $k$-lane partition $\mathcal{P}'$ of $\mathcal{I}'$ such that $G$ is the completion of $(G', \mathcal{I}', \mathcal{P}')$.
\end{enumerate}

\subparagraph{\Cref{item1} implies \Cref{item2}.} Assuming \Cref{item1}, we construct the required $(G', \mathcal{I}', \mathcal{P}')$ for \Cref{item2} as follows. We write $X$ to denote the total number of operations.
For each vertex $v \in V$, we select its interval $I_v=[L_v, R_v]$ by the following rules.
\begin{itemize}
   \item If $v$ belongs to the initial path $P = (\tau_1, \ldots, \tau_k)$, then $L_v = 0$.
    \item If $v$ is created in the $x$th operation, then $L_v = x$.
    \item If the $x$th operation causes $v$ to become a non-designated vertex, then $R_v = x-1$.
    \item If $v$ is still a designated vertex by the end of the construction, then $R_v = X$.
\end{itemize}
Intuitively, $I_v=[L_v, R_v]$ is the time interval where $v$ is a designated vertex during the construction.
This interval representation $\mathcal{I}'$ naturally admits a $k$-lane partition $\mathcal{P}'$ where $v$ belongs to the $i$th lane if $v$ is the $i$th designated vertex at the moment $v$ is added to the graph in the construction.
Let $G'$ be the subgraph of $G$ induced by the edges created by $\inse$, so  $\mathcal{I}'$ is indeed an interval representation of $G'$, and the completion of $(G', \mathcal{I}', \mathcal{P}')$ adds precisely the edges created by $\insv$ and the edges in the initial path $P$.

\subparagraph{\Cref{item2} implies \Cref{item1}.} Assuming \Cref{item2}, we show a sequence of operations required for \Cref{item1} as follows. The initial path $P = (\tau_1, \ldots, \tau_k)$ is chosen as the path resulting from concatenating the initial vertices of all lanes in the completion of $(G', \mathcal{I}', \mathcal{P}')$. Afterward, we go from left to right, adding vertices and edges using $\insv$ and $\inse$ whenever we encounter a new interval, as follows. We sort the vertex set $V \setminus \{\tau_1, \ldots, \tau_k\}$ and the edge set $E'$ of $G'$ \emph{together}, where the value of a vertex $v \in V$ is $L_v = \min I_v$ and the value of an edge $e=\{u,v\} \in E$ is $\max \{L_u, L_v\} = \min I_u \cap I_v$. When there is a tie, we prioritize vertices over edges.
We add vertices and edges according to the sorted order. 
\begin{itemize}
    \item When a vertex $v$ is processed, we perform $\insv(i)$, where $i$ is the lane number of $v$ in $\mathcal{P}'$. Observe that the edge $\{v, \tau_i\}$ added in $\insv(i)$ exists in the completion of $(G', \mathcal{I}', \mathcal{P}')$ because $\tau_i$ is the vertex right before $v$ in the $i$th lane.
    \item When an edge $e=\{u,v\}$ is processed, we perform $\inse(i,j)$, where $i$ is the lane number of $u$ in $\mathcal{P}'$ and $j$ is the lane number of $v$ in $\mathcal{P}'$. Observe that the edge $\{\tau_i, \tau_j\}$ added in $\inse(e)$  is exactly the edge $e=\{u,v\}$, as $u=\tau_i$ and $v=\tau_j$, because both $u$ and $v$ are before $e$ in the sorted order, and the vertices after $u$ and $v$ in their respective lanes are after $e$ in the sorted order.
\end{itemize}
Thus, the graph formed by this construction is precisely the completion $G$ of $(G', \mathcal{I}', \mathcal{P}')$.
\end{proof}

\subsection{Recursive graphs with \texorpdfstring{$k$}{k} lanes}\label{sect:recursive_merge}

Given an integer $k > 0$, we consider the class of \emph{$k$-lane graphs}, defined as follows.

\begin{definition}\label{def:lane_graph}
    A \ul{$k$-lane graph} is a graph $G=(V,E)$ with a non-empty set $\mathcal{T}(G) \subseteq [k]$ and two injective functions $\fin: \mathcal{T}(G) \rightarrow V$ and $\fout: \mathcal{T}(G) \rightarrow V$. For each $i \in \mathcal{T}(G)$, we write $\tin_i(G) = \fin(i)$ to denote the $i$th \ul{in-terminal} of $G$ and write $\tout_i(G) = \fout(i)$ to denote the $i$th \ul{out-terminal} of $G$.  
\end{definition}

Intuitively, $\mathcal{T}(G) \subseteq [k]$ indicates a subset of lanes used by $G$, where $\tin_i(G)$ is the left-most vertex in the $i$th lane, and $\tout_i(G)$ is the right-most vertex in the $i$th lane. We allow $\tin_i(G)=\tout_i(G)$.
We consider two ways of combining $k$-lane graphs. See \Cref{fig:f8} for an illustration.

\begin{figure}[ht!]
    \centering
    \includegraphics[scale=0.65]{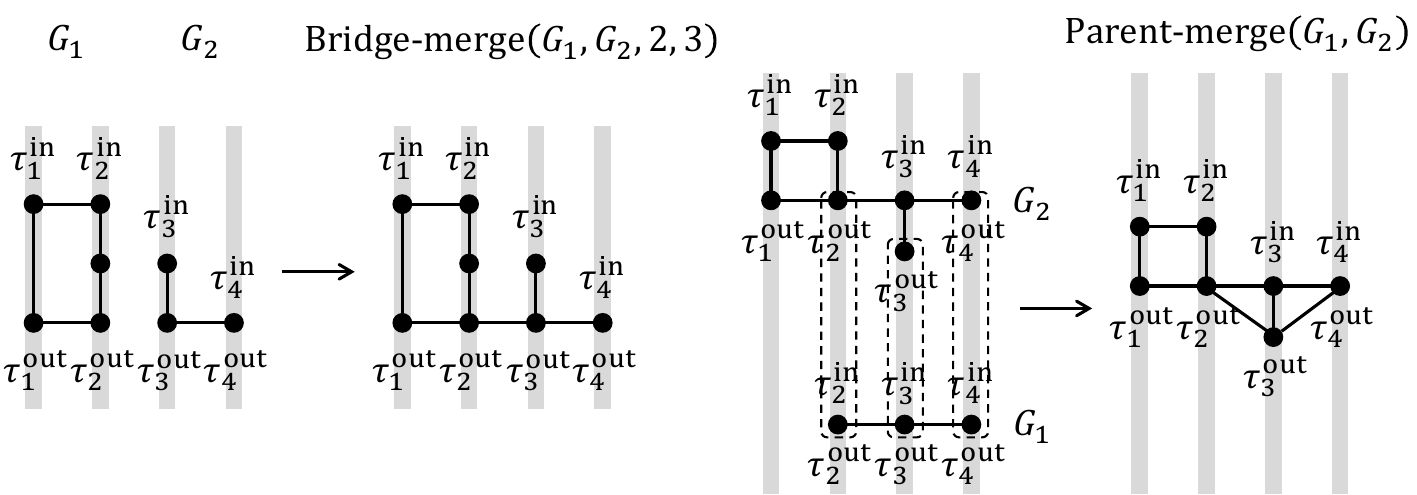}
    \caption{$\bmerge$ and $\pmerge$}
    \label{fig:f8}
\end{figure}

\subparagraph{$\bmerge$.} Suppose $\mathcal{T}(G_1) \cap\mathcal{T}(G_2) = \emptyset$, $i \in \mathcal{T}(G_1)$, and $j \in \mathcal{T}(G_2)$. We define $\bmerge(G_1,G_2,i,j)$ as the $k$-lane graph $G=(V,E)$ with the vertex set $V= V(G_1) \cup V(G_2)$ and the edge set $E = E(G_1) \cup E(G_2) \cup \{e\}$, where $e=\left\{\tout_i(G_1), \tout_j(G_2)\right\}$. We set $\mathcal{T}(G) = \mathcal{T}(G_1) \cup \mathcal{T}(G_2)$. The in-terminals and the out-terminals of $G$ are inherited from $G_1$ and $G_2$, as follows.
\begin{align*}
\tin_i(G) &=\begin{cases}
        \tin_i(G_1) & \text{if $i \in \mathcal{T}(G_1)$,} \\
        \tin_i(G_2) & \text{if $i \in \mathcal{T}(G_2)$.}
    \end{cases} &
    \tout_i(G)&=\begin{cases}
        \tout_i(G_1) & \text{if $i \in \mathcal{T}(G_1)$,} \\
        \tout_i(G_2) & \text{if $i \in \mathcal{T}(G_2)$.}
    \end{cases}
\end{align*}
In other words, given two $k$-lane graphs $G_1$ and $G_2$ on disjoint sets of lanes, $\bmerge(G_1,G_2,i,j)$ is the result of combining $G_1$ and $G_2$ by adding an edge between the $i$th out-terminal of $G_1$ and the $j$th out-terminal of $G_2$.

\subparagraph{$\pmerge$.} Suppose $\mathcal{T}(G_1) \subseteq \mathcal{T}(G_2)$. We define $\pmerge(G_1,G_2)$ as the $k$-lane graph $G=(V,E)$ resulting from identifying $\tin_i(G_1)$ with $\tout_i(G_2)$ for each $i \in \mathcal{T}(G_1)$. We set $\mathcal{T}(G) = \mathcal{T}(G_2)$. The in-terminals and the out-terminals of $G$ are defined as follows.
\begin{align*}
\tin_i(G) &= \tin_i(G_2). &
    \tout_i(G)&=\begin{cases}
        \tout_i(G_1) & \text{if $i \in \mathcal{T}(G_1)$,} \\
        \tout_i(G_2) & \text{if $i \in \mathcal{T}(G_2) \setminus \mathcal{T}(G_1)$.}
    \end{cases}
\end{align*}
Moreover, we have a requirement that identifying $\tin_i(G_1)$ with $\tout_i(G_2)$ for each $i \in \mathcal{T}(G_1)$ does not result in identifying an edge in $G_1$ with an edge $G_2$. In other words, $E$ is a \emph{disjoint} union of the edge set of $G_1$ and the edge set of $G_2$.

In other words, given two $k$-lane graphs $G_1$ and $G_2$ where the set of lanes of $G_1$ is a subset of the set of lanes of $G_2$, $\pmerge(G_1, G_2)$ is the result of combining $G_1$ and $G_2$ by identifying every in-terminal of $G_1$ with the out-terminal of $G_2$ in the same lane.

Similar to \Cref{terminal-recursive-gluing}, we define \emph{$k$-lane recursive graphs} as follows.

\begin{definition}
    \label{terminal-recursive-gluing2}
 Let $B$ denote the set of all $k$-lane graphs with at most $k$ vertices. The set of \ul{$k$-lane recursive graphs} is defined as the closure of $B$ under $\bmerge$ and $\pmerge$.
\end{definition}

\subsection{Hierarchical decompositions with bounded depth}\label{sect:recursive_def}

\Cref{terminal-recursive-gluing2} alone does not lead to an $O(\log n)$-bit proof labeling scheme, as the depth of recursion of a $k$-lane recursive graph can be unbounded. To reduce the depth of recursion, we consider the following operation. See \Cref{fig:f9} for an illustration.

\begin{figure}[ht!]
    \centering
    \includegraphics[scale=0.65]{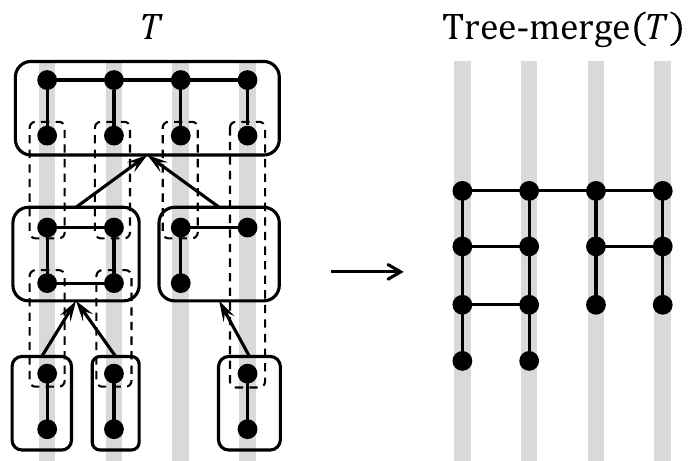}
    \caption{$\tmerge$}
    \label{fig:f9}
\end{figure}

\subparagraph{$\tmerge$.} Suppose $T$ is a rooted tree on graphs such that any $G_1 \in V(T)$ and $G_2 \in V(T)$ satisfy the following conditions:
\begin{itemize}
    \item If $G_1$ is a child of $G_2$ in $T$, then $\mathcal{T}(G_1) \subseteq \mathcal{T}(G_2)$.
    \item If $G_1$ and $G_2$ have the same parent in $T$, then $\mathcal{T}(G_1) \cap \mathcal{T}(G_2) = \emptyset$.  
\end{itemize}
We define $\tmerge(T)$ as the $k$-lane graph $G$ resulting from the following construction:
\begin{enumerate}
    \item Start from the tree $T$.
    \item While there still exists an edge $e=\{G_1, G_2\}$ in the tree, where $G_1$ is a child of $G_2$, contract the edge $e=(G_1, G_2)$ into $\pmerge(G_1,G_2)$.
\end{enumerate}

Observe that $\pmerge$ is \emph{associative} in that the ordering of the edges in the while-loop does not affect the outcome of $\tmerge(T)$, which is the result of the following operations: For each graph $G'$ in $T$ and each of its children $G''$, identify each in-terminal $\tin_i(G'')$ of $G''$ with the out-terminal $\tout_i(G')$ of $G'$ in the same lane $i$.

We allow $T$ to be a single-vertex tree, in which case $\tmerge(T) = G'$, where $G'$ is the unique vertex of $T$. Intuitively, as multiple uses of $\pmerge$ can be condensed into one $\tmerge$, using $\tmerge$ instead of $\pmerge$ allows us to reduce the depth of recursion significantly in the construction of a $k$-lane recursive graph.

\subparagraph{Five types of nodes.}
Let us focus on a specific way to construct $k$-lane recursive graphs using the following five types of nodes, where $\vnode$, $\enode$, and $\pnode$ capture the base cases, $\bnode$ captures the $k$-lane graphs constructed by $\bmerge$, and $\tnode$ captures the $k$-lane graphs constructed by $\tmerge$.

\begin{description}
    \item[{$\vnode$}:] A $\vnode$ is a single-vertex $k$-lane graph $G$ with $\mathcal{T}(G)=\{i\}$ for some $i \in [k]$. Trivially, we must have $\tin_i(G)=\tout_i(G)=v$, where $v$ is the only vertex in $G$.
    \item[{$\enode$}:] An $\enode$ is a single-edge $k$-lane graph $G$ with $\mathcal{T}(G)=\{i\}$ for some $i \in [k]$. It is required that  $\tin_i(G) \neq \tout_i(G)$ are the two endpoints of the only edge $e$ in $G$.
    \item[{$\pnode$}:]  A $\pnode$ is a $k$-lane graph $G$ that is a $k$-vertex path with $\mathcal{T}(G)= [k]$. It is required that for each $i \in \mathcal{T}$, $\tin_i(G) = \tout_i(G)$ is the $i$th vertex of the path.
    \item[{$\bnode$}:] A $\bnode$ is a $k$-lane graph $\bmerge(G_1,G_2,i,j)$ for some indices $i$ and $j$ such that, for both $\ell\in\{1,2\}$, $G_\ell$ is one of $\{\vnode,\tnode\}$.
    \item[{$\tnode$}:] A $\tnode$ is a $k$-lane graph $\tmerge(T)$ for some tree $T$ such that each $G' \in V(T)$ is one of $\{\enode, \pnode, \bnode\}$. 
\end{description}

We remark that the number of node types can be reduced by one, as a $\pnode$ can be constructed from a collection of $k$ isolated vertices through a series of $\bmerge$ operations. However, we include $\pnode$ in the definition for convenience, as it allows us to directly capture the initial path $P = (\tau_1, \ldots, \tau_k)$ in the definition of lanewidth in \Cref{lanewidth-defn}.


\subparagraph{Hierarchical decompositions} Any $k$-lane recursive graph $G$ constructed in this way can be described by a \emph{hierarchical decomposition} $\mathcal{H}$, which is a tree rooted at $G$ formed as follows. For each $\bnode$ $G^\ast = \bmerge(G_1,G_2,i,j)$ in the construction, we set $G_1$ and $G_2$ to be the children of $G^\ast$. For each $\tnode$  $G^\ast = \tmerge(T)$ in the construction, we set all $G' \in V(T)$ to be the children of $G^\ast$. We make a key observation that the hierarchical decomposition $\mathcal{H}$ has bounded depth.

\begin{observation}\label{obs:tree_depth}
Each root-to-leaf path in a hierarchical decomposition $\mathcal{H}$ with parameter $k$ contains at most $2k$ nodes.
\end{observation}
\begin{proof}
We first show that any root-to-leaf path $P = (Q_1, Q_2, \ldots, Q_\ell)$ in $\mathcal{H}$ contains at most $k-1$ $\bnode$s. By the definition of $\tmerge$ and $\bmerge$, we have \[[k] \supseteq \mathcal{T}(Q_1) \supseteq \mathcal{T}(Q_2) \supseteq \cdots \supseteq  \mathcal{T}(Q_\ell) \neq \emptyset \]
Moreover, by the definition of $\bmerge$, if $G^\ast = \bmerge(G_1,G_2,i,j)$, then $\mathcal{T}(G_1)$ and $\mathcal{T}(G_2)$ must be disjoint non-empty subsets of $\mathcal{T}(G^\ast)$, so $1 \leq |\mathcal{T}(G_1)| < |\mathcal{T}(G^\ast)|$ and $1 \leq |\mathcal{T}(G_2)| < |\mathcal{T}(G^\ast)|$. Therefore, we can afford to have at most $k-1$ $\bnode$s in $P$, since otherwise $|\mathcal{T}(Q_\ell)| = 0$.

Except for the leaf node, $\tnode$ and $\bnode$ must appear alternately in a root-to-leaf path. Therefore, $\ell-1 \leq 2\cdot(k-1)+1$, so $\ell \leq 2k$.
\end{proof}

\Cref{obs:tree_depth} plays a key role in designing $O(\log n)$-bit proof labeling schemes, as it implies the following \emph{congestion} bound. Observe that both $\bmerge$ and $\tmerge$ do not merge edges, so each edge $e$ in the original graph $G$ can only appear in one root-to-leaf path in the hierarchical decomposition $\mathcal{H}$, which contains at most $2k$ nodes by \Cref{obs:tree_depth}. 

In addition to the constant depth bound established in \Cref{obs:tree_depth}, our decomposition has the desirable property that the subgraph associated with each node in the decomposition is \emph{connected}. As a result, certifying a node requires only the edges within the subgraph corresponding to that node. This stands in contrast to path and tree decompositions, where each part of the decomposition may correspond to a disconnected subgraph.

\subsection{Construction}\label{sect:recursive_construct}

We show that any graph of lanewidth $k$ can be constructed as a $\tnode$ with parameter $k$.

\begin{proposition}\label{prop:recursive}
Any graph $G$ of lanewidth $k$ can be constructed as a $\tnode$ with parameter $k$.
\end{proposition}
\begin{proof}
By \Cref{lanewidth-defn}, $G$ can be constructed from a $k$-vertex path $P = (\tau_1, \ldots, \tau_k)$ using the two operations $\insv$ and $\inse$. Let $X$ be the total number of operations in the construction. Let $G_0 = P$. For each $x \in [X]$, let $G_x$ be the result of applying the first $x$ operations, so $G = G_X$. We will prove by an induction on $x$ that $G_x$ can be constructed as a $\tnode$ $G_x = \tmerge(T)$ for some tree $T$ with $\mathcal{T}(G_x) = [k]$ and $\tout_i(G_x)=\tau_i$ for each $i \in [k]$, where $\tau_i$ is the $i$th designated vertex of $G_i$. See \Cref{fig:f10} for an example of such a construction.

\begin{figure}[ht!]
    \centering
    \includegraphics[scale=0.65]{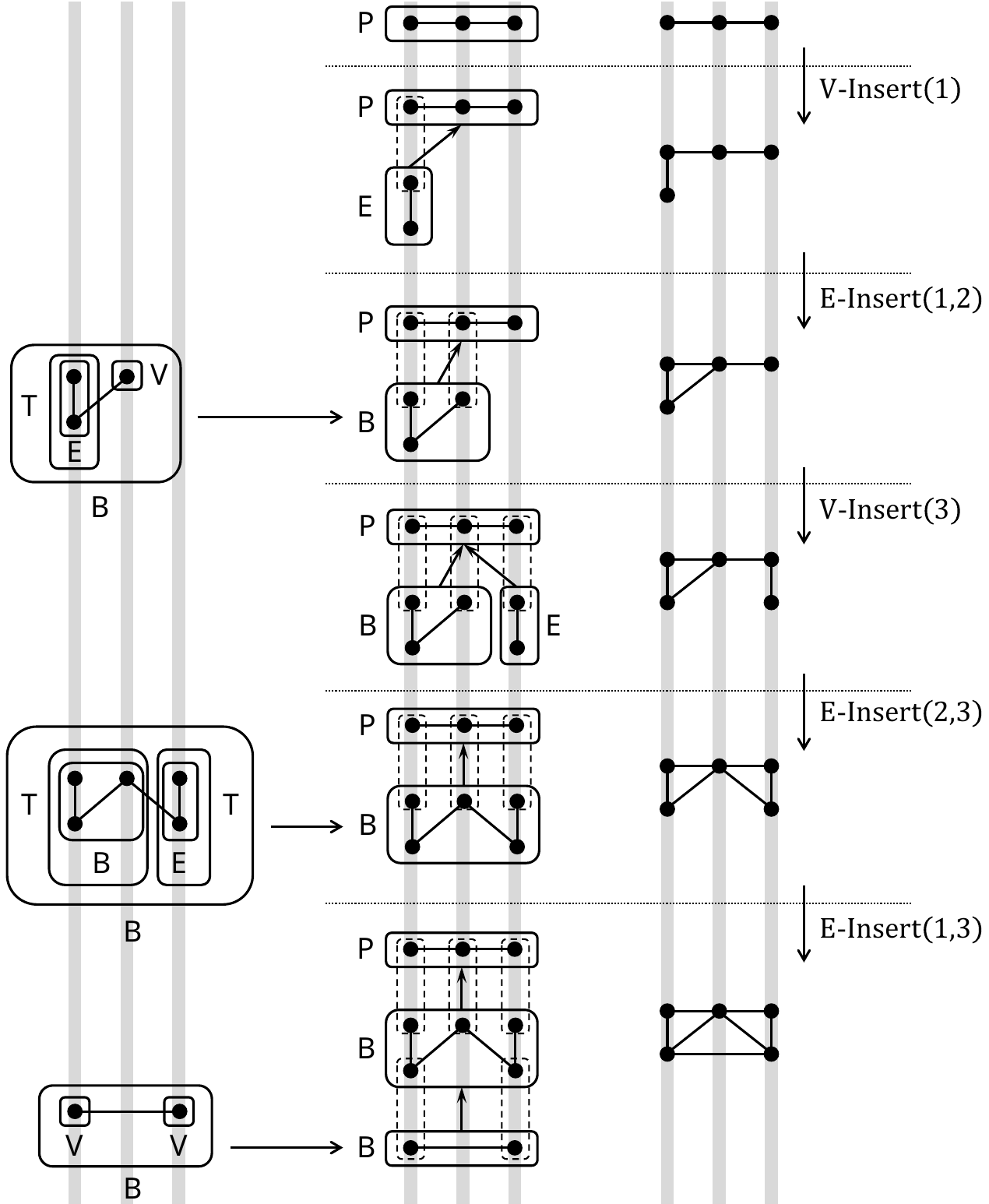}
    \caption{Constructing a bounded-lanewidth graph as a  $\tnode$.}
    \label{fig:f10}
\end{figure}

\subparagraph{Base case.} At the start of the construction, we create a $\pnode$ from the initial path $G^\ast = P = (\tau_1, \ldots, \tau_k)$ by setting $\mathcal{T}(G^\ast) = [k]$ and $\tin_i(G^\ast)=\tout_i(G^\ast)=\tau_i$ for each $i \in [k]$. This $\pnode$ is then encapsulated within a $\tnode$ $G_0 = \tmerge(T)$, where $T$ is the trivial single-vertex tree consisting solely of the $\pnode$ $G^\ast$.

\subparagraph{Inductive step.} Now we assume that $G_{x-1}$ can be constructed as a $\tnode$ $G_{x-1} = \tmerge(T)$ with $\mathcal{T}(G_{x-1}) = [k]$ and $\tout_i(G_{x-1})=\tau_i$ for each $i \in [k]$, where $\tau_i$ is the $i$th designated vertex of $G_{x-1}$. We modify the tree $T$ according to the $x$th operation.

\subparagraph{Case 1.} Suppose the $x$th operation is $\insv(i)$, where we add a vertex $v$ with an edge $e = \{v, \tau_i\}$, and reset the $i$th designated vertex to $v$. We create an $\enode$ $G^\ast$ consisting only of this edge $e$, with $\mathcal{T}(G^\ast) = \{i\}$, $\tin_i(G^\ast) = v$, and $\tin_i(G^\ast) = \tau_i$. This $\enode$ is then added to $T$ as a leaf by making $G^\ast$ a child of $G'$, where $G'$ is the lowest node in $V(T)$ that contains the $i$th out-terminal $\tau_i$ of $G_{x-1}$. Observe that $\pmerge(G^\ast, G')$ resets the $i$th out-terminal from $\tau_i$ to $v$, as required.

\subparagraph{Case 2.} Suppose the $x$th operation is $\inse(i,j)$, where we add an edge $e = \{\tau_i, \tau_j\}$. Let $G_i$ be the lowest node of $V(T)$ containing $\tau_i$. Let $G_j$ be the lowest node of $V(T)$ containing $\tau_j$. Let $G'$ be the lowest common ancestor of $G_i$ and $G_j$ in $T$. 
We will create a $\bnode$ $G^\ast=\bmerge(G_i^\ast, G_j^\ast, i, j)$, for some choices of $G_i^\ast$ and $G_j^\ast$, that uses the new edge $e$ to connect two parts of the graph, and then add the $\bnode$ $G^\ast$ as a child of $G'$ in $T$. There are three sub-cases.

\subparagraph{Case 2.1.} Suppose $G_i = G_j = G'$.  For both $\ell \in \{i,j\}$, let $G_\ell^\ast$ be a $\vnode$ for the vertex $\tau_\ell$ with $\mathcal{T}(G_\ell^\ast)=\{\ell\}$.  We create a $\bnode$ $G^\ast =\bmerge(G_i^\ast, G_j^\ast, i, j)$ to cover the new edge $e = \{\tau_i, \tau_j\}$. We add the $\bnode$ $G^\ast$ as a child of $G_i = G_j = G'$ in $T$.

\subparagraph{Case 2.2.} Suppose $G_i \neq G'$ and $G_j \neq G'$. For both $\ell \in \{i,j\}$, let $T_\ell$ be the subtree of $T$ rooted at the child of $G'$ that is an ancestor of $G_\ell$, and let $G_\ell^\ast = \tmerge(T_\ell)$ be a $\tnode$ for $T_\ell$.  We create a $\bnode$ $G^\ast =\bmerge(G_i^\ast, G_j^\ast, i, j)$ to cover the new edge $e = \{\tau_i, \tau_j\}$, remove $T_i$ and $T_j$ from $T$, and add the $\bnode$ $G^\ast$ as a child of $G'$ in $T$.

\subparagraph{Case 2.3.} Suppose $G'$ is exactly one of $G_i$ and $G_j$. By symmetry, we only consider the case where $G' = G_j$, meaning that $G_i$ is a descendant of $G_j$ in $T$. This case combines the approaches in the above two cases. We create a $\vnode$ $G_j^\ast$ for $\tau_j$ with $\mathcal{T}(G_j^\ast)=\{j\}$. We create a $\tnode$ $G_i^\ast = \tmerge(T_i)$ for $T_i$, which is the subtree of $T$ rooted at the unique child of $G'=G_j$ that is an ancestor of $G_i$. We create a $\bnode$ $G^\ast =\bmerge(G_i^\ast, G_j^\ast, i, j)$ to cover the new edge $e = \{\tau_i, \tau_j\}$, remove $T_i$ from $T$, and add the $\bnode$ $G^\ast$ as a child of $G'$ in $T$.
\end{proof}

\section{Certifying \texorpdfstring{$k$}{k}-lane recursive graphs}\label{sect:certification}

In this section, we prove \Cref{thm:main}.
By \Cref{low-completion-cost,lanewidth-lem}, it suffices to design an $O(\log n)$-bit proof labeling scheme for a given $\MSO$ graph property $\phi$ in bounded-lanewidth graphs. In \Cref{sect:classes}, we establish an analogue of \Cref{homomorphism-class-mapping} for $k$-lane recursive graphs. This result, combined with \Cref{obs:tree_depth}, enables us to construct the desired proof labeling scheme in \Cref{subsect:certification}.

\subsection{Homomorphism classes}\label{sect:classes}

We extend \Cref{homomorphism-class-mapping} to the class of $k$-lane recursive graphs.

\begin{proposition}
    \label{homomorphism-class-mapping2}
    Let $F_k^\ast$ be the set of $k$-lane recursive graphs, and $\phi$ be an $\MSO$ property, then there exists a function $h^\ast : F_k^\ast \to C^\ast$, for some finite set $C^\ast$ of \ul{homomorphism classes} satisfying the following conditions.
    \begin{itemize}
        \item Given any two graphs $G_1, G_2 \in F_k^\ast$ such that $h^\ast(G_1) = h^\ast(G_2)$, $G_1$ satisfies $\phi$ if and only if $G_2$ satisfies $\phi$.
        \item There exist two functions $\fb$ and $\fp$ such that 
\begin{align*}
     h^\ast(\bmerge(G_1, G_2, i, j)) &= \fb(h^\ast(G_1), h^\ast(G_2), i, j),\\
    h^\ast(\pmerge(G_1, G_2)) &= \fp(h^\ast(G_1), h^\ast(G_2)).
\end{align*}
    \end{itemize}
\end{proposition}
\begin{proof}
The proof follows naturally by interpreting $k$-lane recursive graphs as $3k$-terminal recursive graphs and applying \Cref{homomorphism-class-mapping}. The need for $3k$ terminals instead of $2k$ arises from a technical constraint: The composition operator, as defined in \Cref{terminal-recursive-gluing}, does not allow the vertex resulting from merging two terminals to be directly designated as a non-terminal. Consequently, when gluing two terminals $u$ and $v$ into a non-terminal $w$, we must first temporarily assign $w$ as a terminal before applying an additional composition operation to reclassify it as a non-terminal. There can be up to $k$ such temporary terminals in $\pmerge$. In the following discussion, we formalize this proof idea.

\subparagraph{Canonical mapping.} We consider the following canonical way to map the set $S$ of in-terminals and out-terminals in a $k$-lane graph $G=(V,E)$ to $[s] \subseteq [3k]$, where $s = |S| \leq 2k < 3k$ is the number of in-terminals and out-terminals. For any two distinct vertices $u$ and $v$ in $S$, we write $u \prec v$ if one of the following conditions holds:
\begin{itemize}
    \item The lane number of $u$ is smaller than the lane number of $v$.
    \item $u$ and $v$ are in the same lane where $u$ is the in-terminal and $v$ is the out-terminal.
\end{itemize}
For each $v \in S$, we write $\xi(v)$ to denote the rank of $v$ in $S$ with respect to $\prec$. Now $G$ is seen as a $3k$-terminal graph by treating each $v \in S$ as the $\xi(v)$th terminal. The function $\xi\circ \fin:\mathcal{T}(G)\rightarrow[3k]$ maps each lane of $G$ to the rank of the in-terminal of the lane. Similarly, the function $\xi\circ \fout:\mathcal{T}(G)\rightarrow[3k]$ maps each lane of $G$ to the rank of the out-terminal of the lane, where $\fin$ and $\fout$ are defined in \Cref{def:lane_graph}. 

Observe that the two functions $\xi\circ \fin$ and $\xi\circ \fout$ are uniquely determined once we know, for each lane $i \in \mathcal{T}(G)$, whether $\tin_i(G)=\tout_i(G)$, and vice versa. 

\subparagraph{Homomorphism class.} We define the homomorphism class of a $k$-lane recursive graph $G$ as 
\[h^\ast(G) = (h(G), \mathcal{T}(G), \xi\circ \fin, \xi\circ \fout),\]
where $h(G)$ is the homomorphism class of $G$ in \Cref{homomorphism-class-mapping} by viewing $G$ as a $3k$-terminal recursive graph using the aforementioned mapping.

\subparagraph{Validity.}
For any two $k$-lane recursive graphs $G_1$ and $G_2$ with $h^\ast(G_1) = h^\ast(G_2)$, they also have $h(G_1) = h(G_2)$, so $G_1$ satisfies $\phi$ if and only if $G_2$ satisfies $\phi$ by \Cref{homomorphism-class-mapping}.

\subparagraph{$\pmerge$.} We show that, for any $G = \pmerge(G_1, G_2)$,  $h^\ast(G)$ can be inferred from $h^\ast(G_1)$ and $h^\ast(G_2)$. This implies the existence of the desired function $\fp$.
\begin{itemize}
\item $\mathcal{T}(G) = \mathcal{T}(G_2)$ by the definition of $\pmerge$.
\item The two functions $\xi\circ \fin$ and $\xi\circ \fout$ associated with $G$ can be determined from the two functions associated with $G_1$ and $G_2$ by the following rule: For each lane of $G$, the in-terminal and the out-terminal are different vertices if and only if for at least one of $G_1$ and $G_2$, the in-terminal and the out-terminal are different vertices in the same lane.
    \item To determine $h(G)$, we view $G = \pmerge(G_1, G_2)$ as the result of the following procedure, where $G$, $G_1$, and $G_2$ are viewed as $3k$-terminal recursive graphs.
    \begin{enumerate}
        \item Let $\tilde{G} = \pmerge(G_1, G_2)$ where each vertex resulting from merging two terminals is still considered a terminal temporarily. The graph $\tilde{G}$ is viewed as a $3k$-terminal recursive graph resulting from the composition operation corresponding to the merging operation $G = \pmerge(G_1, G_2)$.
        \item Let $G$ be $3k$-terminal recursive graph obtained by retaining only the in-terminals and out-terminals of $\tilde{G}$ as terminals. The graph $G$ can be obtained by composing $G$ with a trivial single-vertex graph, where all temporary terminals in $\tilde{G}$ are reclassified as non-terminals during the composition.
    \end{enumerate}
        The two composition operators are determined by $\mathcal{T}(G)$, $\mathcal{T}(G_1)$, $\mathcal{T}(G_2)$, and the two functions $\xi\circ \fin$ and $\xi\circ \fout$ associated with these three graphs $G$, $G_1$, and $G_2$. 
    Given the two composition operators, $h(G)$ is determined by $h(G_1)$ and $h(G_2)$ in view of \Cref{homomorphism-class-mapping}.
\end{itemize}

\subparagraph{$\bmerge$.} We show that, for any $G = \bmerge(G_1, G_2, i, j)$,  $h^\ast(G)$ can be inferred from $h^\ast(G_1)$, $h^\ast(G_2)$, and $(i,j)$. This implies the existence of the desired function $\fb$.
\begin{itemize}
\item $\mathcal{T}(G) = \mathcal{T}(G_1) \cup \mathcal{T}(G_2)$ by the definition of $\bmerge$.
\item The two functions $\xi\circ \fin$ and $\xi\circ \fout$ associated with $G$ can be determined from the two functions associated with $G_1$ and $G_2$ by the following rule: For each lane of $G$, the in-terminal and the out-terminal are different vertices if and only if the same holds for the one of $G_1$ and $G_2$ where the lane appears.
    \item To determine $h(G)$, we view $G = \bmerge(G_1, G_2, i, j)$ as the result of the following procedure, where $G$, $G_1$, and $G_2$ are viewed as $3k$-terminal recursive graphs.
    \begin{enumerate}
        \item Let $\tilde{G}$ be the disjoint union of $G_1$ and $G_2$, where the set of terminals of $\tilde{G}$ is the union of the terminals of $G_1$ and $G_2$.
        \item Let $G$ be the composition of $\tilde{G}$ and a graph consisting of a single edge $e$, where we identify one endpoint of $e$ with $\tout_i(G_1)$ and the other endpoint of $e$ with  $\tout_j(G_2)$.
    \end{enumerate}
    Similarly, the two composition operators are determined by $(i,j)$,  $\mathcal{T}(G)$, $\mathcal{T}(G_1)$, $\mathcal{T}(G_2)$, and the two functions $\xi\circ \fin$ and $\xi\circ \fout$ associated with these three graphs $G$, $G_1$, and $G_2$. 
    Given the two composition operators, $h(G)$ is determined by $h(G_1)$ and $h(G_2)$ in view of \Cref{homomorphism-class-mapping}.    \qedhere
\end{itemize}
\end{proof}

Same as \Cref{homomorphism-class-mapping}, \Cref{homomorphism-class-mapping2} also applies to graphs where vertices and edges are labeled with inputs from a fixed finite set. 

\subsection{Certification algorithm}\label{subsect:certification}

The main goal of the section is to prove the following result.

\begin{proposition}\label{prop:certification_algo}
For any integer $k \geq 1$, for any $\MSO$ graph property $\phi$, there exists a proof labeling scheme for $\phi$ with $O(\log n)$-bit edge labels in any graph $G$ of lanewidth $k$.
\end{proposition}

Same as Courcelle's theorem, \Cref{prop:certification_algo} applies to graphs with a specified vertex subset or an edge set, or more generally, to graphs where vertices and edges are labeled with inputs from a fixed finite set. We first prove our main theorem using \Cref{prop:certification_algo}.

\mainthm*

\begin{proof}
Suppose the input graph $G=(V,E)$ satisfies the graph property $\phi \wedge (\operatorname{pathwidth}\leq k)$.
By \Cref{pathwidth-defn}, $G$ admits an interval representation $\mathcal{I}$ of width at most $k+1$.
By \Cref{low-completion-cost}, there exists an $s$-lane partition $\mathcal{P}$ with $s \leq f(k+1)$ such that the completion $G'=(V,  E')$ of $(G,\mathcal{I},\mathcal{P})$ can be embedded into $G$ with congestion at most $h(k+1)$. By \Cref{lanewidth-lem}, the lanewidth of $G'$ is at most $f(k+1)$.

We apply \Cref{prop:certification_algo} to $G'$, whose lanewidth is bounded, to certify the $\MSO$ graph property $\phi \wedge (\operatorname{pathwidth}\leq k)$ for the subgraph $G=(V,E)$ of $G'=(V,E')$, by viewing $E \subseteq E'$ as a specified edge subset given as the input to the proof labeling scheme. The proof labeling scheme uses $O(\log n)$-bit edge labels in $G'$. 

To complete the proof, it remains to show that the proof labeling scheme of $G'$ can be simulated using $O(\log n)$-bit edge labels in $G$. By \Cref{prop:degeneracy}, these edge labels can subsequently be simulated using $O(\log n)$-bit vertex labels, as graphs of bounded pathwidth have bounded degeneracy.

\subparagraph{Certifying the embedding.} For each virtual edge $e=\{u,v\} \in E' \setminus E$ in $G'$ that is embedded as a path $P_e = (u, \ldots, v)$, we store the following $O(\log n)$-bit information in each edge $e'$ of $P_e$.
\begin{itemize}
    \item $\ID(u)$ and $\ID(v)$. This indicates that $e'$ is in the path $P_e$, which starts from $u$ and ends at $v$.
    \item The $O(\log n)$-bit edge label $\ell(e)$ of $e$ in the proof labeling scheme of $G'$.
    \item The rank of $e'$ in $P_e$ in the forward direction.
    \item The rank of $e'$ in $P_e$ in the backward direction.
\end{itemize}
The validity of the embedding can be certified as follows. Each vertex $x \in V$ checks whether it has at most two incident edges in $P_e$. If $x$ has two incident edges in $P_e$, it identifies itself as an intermediate vertex of $P_e$ and verifies that both edges carry the same label $\ell(e)$. It then checks their ranks to confirm that the two edges appear consecutively in $P_e$. If $x$ has only one incident edge in $P_e$, it identifies itself as an endpoint of $P_e$ and verifies that the rank of the edge is one in exactly one of the two directions. If the rank is one in the forward direction, $x$ verifies that $\ID(x) = \ID(u)$. If the rank is one in the backward direction, $x$ verifies that $\ID(x) = \ID(v)$.

Therefore, to pass all the checks, it is necessary and sufficiently that $P_e$ is a simple path from $u$ to $v$, with all edges carrying the same $O(\log n)$-bit edge label $\ell(e)$ assigned to $e$. The congestion of the embedding is upper bounded by a constant $h(k+1)$, so the simulation incurs an additive cost of only $O(\log n)$ bits per edge in $E$.

\subparagraph{Summary.} If the input graph $G = (V, E)$ satisfies the graph property $\phi \wedge (\operatorname{pathwidth} \leq k)$, then the certification assignment $\mathbf{P}$ described above ensures that the verifier $\mathbf{V}$ accepts everywhere. Conversely, if $G = (V, E)$ does not satisfy the property $\phi \wedge (\operatorname{pathwidth} \leq k)$, the verifier $\mathbf{V}$ must reject somewhere, regardless of the choice of the $O(\log n)$-bit proofs.

It might still be possible to embed a graph $G'$ of lanewidth at most $f(k+1)$ to $G$ with congestion at most $h(k+1)$ even if the pathwidth of $G$ is greater than $k$. In this case, the verifier $\mathbf{V}$ still rejects somewhere: By the correctness of \Cref{prop:certification_algo}, to make everyone accept, it is required that the subgraph $G=(V,E)$ of $G'=(V,E')$ satisfies the graph property $\phi \wedge (\operatorname{pathwidth}\leq k)$.
\end{proof}

In the remainder of the section, we prove \Cref{prop:certification_algo}. By \Cref{prop:recursive}, any graph $G$ of lanewidth $k$ can be constructed as a $\tnode$ with parameter $k$, so it suffices to focus on $k$-lane recursive graphs that can be constructed using the five types of nodes in \Cref{sect:recursive_def}. By \Cref{homomorphism-class-mapping2}, to certify the given $\MSO$ graph property $\phi$ of a $k$-lane recursive graph, it suffices to certify its homomorphism class. 

\subparagraph{Certifying the validity of a $k$-lane recursive graph.} Given an $\MSO$ graph property $\phi$, we define the task of certifying the \emph{validity} of a $k$-lane recursive graph $G=(V,E)$, as follows.
\begin{description}
    \item[Input:] Each edge $e \in E$ is initially provided with the following input data: 
    \begin{itemize}
        \item A set of lanes $\mathcal{T} \subseteq [k]$.
        \item A homomorphism class $c \in C^\ast$ with respect to the $\MSO$ graph property $\phi$.
        \item An assignment that maps each lane $i \in \mathcal{T}$ to an in-terminal and an out-terminal in $V$.
    \end{itemize}
    \item[Output:] If $G$ can be realized as a $k$-lane recursive graph consistent with the provided input data, then the labeling generated by the certificate assignment $\mathbf{P}$ ensures that all vertices accept in the local verification algorithm $\mathbf{V}$. Otherwise, regardless of the labeling, at least one vertex rejects in the local verification algorithm $\mathbf{V}$.
\end{description}

Given the above definition, it is convenient to have the following terminology.

\begin{definition}\label{def:basic}
Given an $\MSO$ graph property $\phi$, the \ul{basic information} $\mathcal{B}(G)$ of a $k$-lane recursive graph $G$ consists of its lane set $\mathcal{T}(G)$, its homomorphism class $h^\ast(G)$ with respect to the $\MSO$ graph property $\phi$, its in-terminals $\{\tin_i(G) \, | \, i\in\mathcal{T}(G)\}$, and its out-terminals  $\{\tout_i(G) \, | \, i\in\mathcal{T}(G)\}$.
\end{definition}

The task of certifying the validity of a $k$-lane recursive graph $G$ can be seen as verifying the correctness of its basic information $\mathcal{B}(G)$, which can be encoded as an $O(\log n)$-bit string.

  We begin with the base case: $\enode$ and $\pnode$. Since we label edges instead of vertices, it is not meaningful to consider $\vnode$. As we will see later, the certification of the validity of a $\vnode$ in a hierarchical decomposition is implicitly included in the certification of $\bnode$.

\begin{lemma}\label{lem:base_case}
For any $k$-lane recursive graph $G=(V,E)$ that can be constructed as an $\enode$ or a $\pnode$, its validity can be certified with $O(\log n)$-bit edge labels. 
\end{lemma}
\begin{proof}
A $\pnode$ consists of a path of $k$ vertices. An $\enode$ consists of a single edge. Therefore, we may use the trivial certificate that encodes the entire graph topology in $O(\log n)$ bits. All edges are given the same certificate. A vertex $v$ rejects \emph{if and only if} it detects at least one of the following.
\begin{itemize}
    \item The edges incident to $v$ have different certificates.
    \item The edges incident to $v$ are inconsistent with the graph topology given in the certificate.
    \item The input data is inconsistent with the graph topology given in the certificate.
\end{itemize}
If $G$ is a desired $k$-lane recursive graph, then all vertices accept. Otherwise, regardless of the assigned certificates, at least one vertex rejects.
\end{proof}

Next, we consider the remaining cases: $\bnode$ and $\tnode$. In these cases, we certify the entire hierarchical decomposition $\mathcal{H}$. By \Cref{obs:tree_depth}, $\mathcal{H}$ is a rooted tree of bounded depth, so each edge is shared by at most a constant number of nodes in $\mathcal{H}$. Hence the congestion contributes only a constant-factor overhead. 

\begin{lemma}\label{lem:inductive_case}
For any $k$-lane recursive graph $G=(V,E)$ that can be constructed as a $\bnode$ or a $\tnode$, its validity can be certified with $O(\log n)$-bit edge labels. 
\end{lemma}
\begin{proof}
 We prove the lemma by an induction on the depth $d$ of the hierarchical decomposition $\mathcal{H}$ associated with the recursive construction of $G$. By \Cref{obs:tree_depth}, $d \leq 2k-1$ is bounded, so the proof by induction does not incur any non-constant overhead in the label size.
 For the base case, the lemma statement vacuously holds when $d = 0$ because the recursion depth must be at least one for $\bnode$ and $\tnode$.

 \subparagraph{$\bnode$.} Consider the case $G=(V,E)$ is a $\bnode$ formed by $G=\bmerge(G_1, G_2, i, j)$. Our certificate for the validity of $G$, which consists of $O(\log n)$-bit edge labels, is as follows.
 \begin{enumerate}
     \item Each edge $e \in E$ is given the following information: $(i,j)$, $\mathcal{B}(G_1)$, and $\mathcal{B}(G_2)$.
     \item \label{item-2} For each $x \in \{1,2\}$, we certify the correctness of $\mathcal{B}(G_x)$ as follows.
     \begin{itemize}
         \item If $G_x$ is a $\tnode$, each edge in $G_x$ is provided a certificate for the correctness of the basic information $\mathcal{B}(G_x)$. By the induction hypothesis, such a certificate uses $O(\log n)$-bit edge labels. 
         \item If $G_x$ is a $\vnode$, we apply the proof labeling scheme of \Cref{vertex-in-graph-edge-pls} pointing to the unique vertex of $G_x$. This requires $O(\log n)$-bit edge labels in $G$ and allows us to certify that $G_x$ consists of a single vertex. The unique vertex of $G_x$ can locally verify the basic information $\mathcal{B}(G_x)$.
     \end{itemize} 
 \end{enumerate} 

Due to \Cref{item-2}, in the subsequent discussion, we assume that $\mathcal{B}(G_1)$ and $\mathcal{B}(G_2)$ have already been certified. The fact that $G=\bmerge(G_1, G_2, i, j)$ can also be certified by letting each vertex locally check if $\left\{\tout_i(G_1),\tout_j(G_2)\right\}$ is the unique edge between $G_1$ and $G_2$. 

Given the correctness of $\mathcal{B}(G_1)$, $\mathcal{B}(G_2)$, and $G=\bmerge(G_1, G_2, i, j)$, the correctness of the basic information  $\mathcal{B}(G)$ can be checked locally in each vertex, as follows.
\begin{itemize}
    \item $\mathcal{T}(G) = \mathcal{T}(G_1) \cup \mathcal{T}(G_2)$.
    \item The homomorphism class $h^\ast(G)$ is determined by $(i,j)$ and the two homomorphism classes $h^\ast(G_1)$ and $h^\ast(G_2)$, according to \Cref{homomorphism-class-mapping2}.
    \item The set of in-terminals of $G$ consists of the in-teriminals of $G_1$ and $G_2$.
    \item The set of out-terminals of $G$ consists of the out-teriminals of $G_1$ and $G_2$.
\end{itemize}



 \subparagraph{$\tnode$.} Consider the case $G=(V,E)$ is a $\tnode$ formed by $G=\tmerge(T)$. Each graph $G' \in V(T)$ is one of $\{\enode,\pnode,\bnode\}$. In our certificate for the correctness of $\mathcal{B}(G)$, we provide the basic information $\mathcal{B}(G')$ to each edge $e$ in $G'$, for all $G' \in V(T)$. This information is then certified as follows.
\begin{itemize}
    \item If $G'$ is an $\enode$ or a $\pnode$, we certify the correctness of $\mathcal{B}(G')$ using $O(\log n)$-bit edge labels in $G'$ by \Cref{lem:base_case}.
    \item If $G'$ is a $\bnode$, we certify the correctness of $\mathcal{B}(G')$ using $O(\log n)$-bit edge labels in $G'$ by induction hypothesis. 
\end{itemize}
Since $\pmerge$ does not combine edges, the graphs $G' \in V(T)$ have disjoint edge sets, so the information assigned to each edge in $G$ so far consists of $O(\log n)$ bits.

 Given that the validity of each $G' \in V(T)$ has been certified, the next task is to certify that $G$ is indeed formed by the operation $G=\tmerge(T)$. Here we utilize the \emph{associativity} of $\pmerge$, meaning that the ordering of the execution of $\pmerge$ operations within a $\tmerge$ operation does not affect the outcome. For each $G' \in V(T)$, we do the following.
 \begin{itemize}
     \item Define $T_{G'}$ as the subtree of $T$ rooted at $G'$.
     \item We give each edge in $G'$ the basic information $\mathcal{B}(\tmerge(T_{G'}))$.
 \end{itemize}
 Our plan is to certify the correctness of $\mathcal{B}(\tmerge(T_{G'}))$ using $O(\log n)$-bit edge labels \emph{only} in $G'$, despite that $\tmerge(T_{G'})$ might contain vertices and edges outside $G'$. 
 
 Let $G_1, \ldots, G_x$ be the children of $G'$ in $T$. Observe that $\tmerge(T_{G'}) = \tmerge(T')$, where $T'$ is defined as a tree of depth one rooted at $G'$ with the children  $\tmerge(T_{G_i})$ for all $i \in [x]$. In other words, $T'$ is the result of replacing $T_{G_i}$ with $\tmerge(T_{G_i})$ for each $i \in [x]$ in $T$. 

 For the edges in $G'$, our certificate for the correctness of $\mathcal{B}(\tmerge(T_{G'}))$ consists of  $\mathcal{B}(\tmerge(T_{G_i}))$ for all $i \in [x]$. To verify the correctness of the operation $\tmerge(T')$, each out-terminal of $G'$ can locally check if it is the right in-terminal of the right graph $\tmerge(T_{G_i})$, and each remaining vertex of $G'$ locally check that it does not have any neighbor outside $G'$. Given the correctness of $\mathcal{B}(G')$ and $\mathcal{B}(\tmerge(T_{G_i}))$ for all $i \in [x]$, along
 with the fact that $\tmerge(T_{G'})$ is indeed formed by the operation $\tmerge(T')$, the correctness of the basic information $\tmerge(T_{G'})$ can be checked locally in each vertex in $G'$, as follows.
\begin{itemize}
    \item The set of lanes, in-terminals, and out-terminals of $\tmerge(T_{G'})$ can be calculated from $\mathcal{B}(G')$ and $\mathcal{B}(\tmerge(T_{G_i}))$ for all $i \in [x]$.
    \item The homomorphism class of $\tmerge(T_{G'})$ is determined by the homomorphism classes of $G'$ and $\tmerge(T_{G_i})$ for all $i \in [x]$ by \Cref{homomorphism-class-mapping2}.
\end{itemize}

 Finally, to certify the correctness of the basic information $\mathcal{B}(G)$ for the underlying graph $G = \tmerge(T)$, we just need to make sure that at least one vertex $v$ in the root $\tilde{G}$ of $T$ exists. If such a vertex exists and all vertices in $G$ accept, then we know that $\mathcal{B}(G) = \mathcal{B}(T_{\tilde{G}})$ is correct. To ensure the existence of such a vertex, we apply the proof labeling scheme of \Cref{vertex-in-graph-edge-pls} pointing to any choice of a vertex $v$ in the root $\tilde{G}$ of $T$.
\end{proof}

We are now prepared to prove \Cref{prop:certification_algo}.

\begin{proof}[Proof of \Cref{prop:certification_algo}]
By \Cref{prop:recursive}, any graph $G$ of lanewidth $k$ can be constructed as a $\tnode$ with parameter $k$. By \Cref{lem:inductive_case}, we can certify its homomorphism class with respect to the $\MSO$ graph property $\phi$ using $O(\log n)$-bit edge labels. By \Cref{homomorphism-class-mapping2}, from the homomorphism class, we can infer whether $G$ has the property $\phi$.
\end{proof}

\section{Conclusions}\label{sect:conclusion}
In this work, we have established a meta-theorem demonstrating the existence of a proof labeling scheme with \emph{optimal} $O(\log n)$-bit vertex labels that can decide any $\MSO$ property for graphs with bounded pathwidth. Pathwidth plays a crucial role in parameterized complexity theory~\cite{downey27082016} and structural graph theory~\cite{kawarabayashi2007some,lovasz2006graph}. The results and techniques shown in this paper represent a step toward a deeper understanding of local certification and local algorithms for various sparse graph classes~\cite{nevsetvril2012sparsity} that are considered in these fields---a research topic that has garnered significant attention in recent years~\cite{bousquet2024local,ESPERET202268,esperet_et_al:LIPIcs.ICALP.2022.58,feuilloley2022can,feuilloley2021compact,feuilloley2023local,fomin2024distributed,fraigniaud2024meta,naor2020power,nevsetvril2016distributed}.

Bounded-pathwidth graphs have a deep connection to the theory of general structure of graph classes characterized by forbidden minors. The Graph Structure Theorem of \citeauthor{robertson2003graph} showed that any such graphs admit a certain structural decomposition which involves path decompositions~\cite{robertson2003graph}. \citeauthor{1530755} developed a polynomial-time algorithm to construct such a decomposition, leading to many applications in approximation algorithms~\cite{1530755}. In the context of local certification, we believe that our certification algorithm for path decompositions will be useful as building blocks to design further local certification algorithms for other graph properties that can be characterized by forbidden minors. One concrete research direction is to explore whether the techniques developed in this work can be extended to cover all graphs with bounded treewidth, thereby achieving a complete improvement over the previous result on local certification on graphs with bounded treewidth~\cite{fraigniaud2024meta}. Moreover, applying the Excluding Grid Theorem of Robertson and Seymour~\cite{ROBERTSON198692}, such an extension would imply the existence of an optimal $O(\log n)$-bit proof labeling scheme for certifying $H$-minor-free graphs for an arbitrary planar graph $H$.

Furthermore, it would be interesting to investigate whether our new techniques, such as $k$-lane recursive graphs, could be applied to study algorithms and complexity for graphs with bounded pathwidth in other contexts.


\printbibliography

@INPROCEEDINGS{1530755,
  author={Demaine, E.D. and Hajiaghayi, M.T. and Kawarabayashi, K.},
  booktitle={46th Annual IEEE Symposium on Foundations of Computer Science (FOCS'05)}, 
  title={Algorithmic graph minor theory: Decomposition, approximation, and coloring}, 
  year={2005},
  volume={},
  number={},
  pages={637-646},
  doi={10.1109/SFCS.2005.14}}

@article{CENSORHILLEL2020112,
title = {Approximate proof-labeling schemes},
journal = {Theoretical Computer Science},
volume = {811},
pages = {112--124},
year = {2020},
note = {Special issue on Structural Information and Communication Complexit},
author = {Censor-Hillel, Keren and Paz, Ami and Perry, Mor},
}

@article{korman03032010,
title = {Proof labeling schemes},
journal = {Distributed Computing},
volume = {22},
pages = {215--233},
year = {2010},
author = {Korman, Amos and Kutten, Shay and Peleg, David},
}

@book{downey27082016,
title = {Fundamentals of Parametrized Complexity},
author = {Downey, Rodney G. and Fellows, Michael R.},
year = {2016},
publisher = {Springer London},
location = {London},
doi = {https://doi.org/10.1007/978-1-4471-5559-1},
isbn = {978-1-4471-7164-5}
}

@article{fraigniaud2024meta,
  title={A meta-theorem for distributed certification},
  author={Fraigniaud, Pierre and Montealegre, Pedro and Rapaport, Ivan and Todinca, Ioan},
  journal={Algorithmica},
  volume={86},
  number={2},
  pages={585--612},
  year={2024},
  publisher={Springer}
}

@article{bousquet2024local,
  title={Local certification of graph decompositions and applications to minor-free classes},
  author={Bousquet, Nicolas and Feuilloley, Laurent and Pierron, Th{\'e}o},
  journal={Journal of Parallel and Distributed Computing},
  volume={193},
  pages={104954},
  year={2024},
  publisher={Elsevier}
}

@article{ESPERET202268,
title = {Local certification of graphs on surfaces},
journal = {Theoretical Computer Science},
volume = {909},
pages = {68--75},
year = {2022},
issn = {0304-3975},
author = {Louis Esperet and Benjamin Lévêque},
}

@InProceedings{emek_et_al:LIPIcs:2020:13098,
  author =	{Yuval Emek and Yuval Gil},
  title =	{{Twenty-Two New Approximate Proof Labeling Schemes}},
  booktitle =	{34th International Symposium on Distributed Computing (DISC 2020)},
  pages =	{20:1--20:14},
  series =	{Leibniz International Proceedings in Informatics (LIPIcs)},
  ISBN =	{978-3-95977-168-9},
  ISSN =	{1868-8969},
  year =	{2020},
  volume =	{179},
  editor =	{Hagit Attiya},
  publisher =	{Schloss Dagstuhl--Leibniz-Zentrum f{\"u}r Informatik},
  address =	{Dagstuhl, Germany},
  URL =		{https://drops.dagstuhl.de/opus/volltexte/2020/13098},
  URN =		{urn:nbn:de:0030-drops-130983},
  doi =		{10.4230/LIPIcs.DISC.2020.20}
}

@article{fraginiaud2019random,
title = {Randomized proof-labeling schemes},
author = {Pierre Fraigniaud and Boaz Patt-Shamir and Mor Perry },
journal = {Distributed Computing},
volume = {32},
pages = {217--234},
year = {2019},
}

@article{ROBERTSON198692,
title = {Graph minors. {V.} Excluding a planar graph},
journal = {Journal of Combinatorial Theory, Series B},
volume = {41},
number = {1},
pages = {92--114},
year = {1986},
issn = {0095-8956},
author = {Neil Robertson and Paul D Seymour},
}

@article{ROBERTSON198339,
title = {Graph minors. {I.} Excluding a forest},
journal = {Journal of Combinatorial Theory, Series B},
volume = {35},
number = {1},
pages = {39--61},
year = {1983},
issn = {0095-8956},
author = {Neil Robertson and Paul D Seymour},
}

@article{borie1992automatic,
  title={Automatic generation of linear-time algorithms from predicate calculus descriptions of problems on recursively constructed graph families},
  author={Borie, Richard B and Parker, R Gary and Tovey, Craig A},
  journal={Algorithmica},
  volume={7},
  pages={555--581},
  year={1992},
  publisher={Springer}
}

@article{v012a019,
 author = {G{\"o}{\"o}s, Mika and Suomela, Jukka},
 title = {Locally Checkable Proofs in Distributed Computing},
 year = {2016},
 pages = {1--33},
 doi = {10.4086/toc.2016.v012a019},
 publisher = {Theory of Computing},
 journal = {Theory of Computing},
 volume = {12},
 number = {19},
 URL = {https://theoryofcomputing.org/articles/v012a019},
}

@article{COURCELLE199012,
title = {The monadic second-order logic of graphs. {I}. Recognizable sets of finite graphs},
journal = {Information and Computation},
volume = {85},
number = {1},
pages = {12-75},
year = {1990},
issn = {0890-5401},
author = {Bruno Courcelle},
}

@InProceedings{esperet_et_al:LIPIcs.ICALP.2022.58,
  author =	{Esperet, Louis and Norin, Sergey},
  title =	{{Testability and Local Certification of Monotone Properties in Minor-Closed Classes}},
  booktitle =	{49th International Colloquium on Automata, Languages, and Programming (ICALP 2022)},
  pages =	{58:1--58:15},
  series =	{Leibniz International Proceedings in Informatics (LIPIcs)},
  ISBN =	{978-3-95977-235-8},
  ISSN =	{1868-8969},
  year =	{2022},
  volume =	{229},
  editor =	{Boja\'{n}czyk, Miko{\l}aj and Merelli, Emanuela and Woodruff, David P.},
  publisher =	{Schloss Dagstuhl -- Leibniz-Zentrum f{\"u}r Informatik},
  address =	{Dagstuhl, Germany},
}

@article{feuilloley2023local,
  title={Local certification of graphs with bounded genus},
  author={Feuilloley, Laurent and Fraigniaud, Pierre and Montealegre, Pedro and Rapaport, Ivan and R{\'e}mila, Eric and Todinca, Ioan},
  journal={Discrete Applied Mathematics},
  volume={325},
  pages={9--36},
  year={2023},
  publisher={Elsevier}
}

@inproceedings{feuilloley2022can,
  title={What can be certified compactly? compact local certification of MSO properties in tree-like graphs},
  author={Feuilloley, Laurent and Bousquet, Nicolas and Pierron, Th{\'e}o},
  booktitle={Proceedings of the 2022 ACM Symposium on Principles of Distributed Computing {(PODC)}},
  pages={131--140},
  year={2022}
}

@InProceedings{fomin2024distributed,
  author =	{Fomin, Fedor V. and Fraigniaud, Pierre and Montealegre, Pedro and Rapaport, Ivan and Todinca, Ioan},
  title =	{{Distributed Model Checking on Graphs of Bounded Treedepth}},
  booktitle =	{38th International Symposium on Distributed Computing (DISC 2024)},
  pages =	{25:1--25:20},
  series =	{Leibniz International Proceedings in Informatics (LIPIcs)},
  ISBN =	{978-3-95977-352-2},
  ISSN =	{1868-8969},
  year =	{2024},
  volume =	{319},
  editor =	{Alistarh, Dan},
  publisher =	{Schloss Dagstuhl -- Leibniz-Zentrum f{\"u}r Informatik},
  address =	{Dagstuhl, Germany},
  URL =		{https://drops.dagstuhl.de/entities/document/10.4230/LIPIcs.DISC.2024.25},
  URN =		{urn:nbn:de:0030-drops-212513},
  doi =		{10.4230/LIPIcs.DISC.2024.25}
}

@inproceedings{naor2020power,
  title={The power of distributed verifiers in interactive proofs},
  author={Naor, Moni and Parter, Merav and Yogev, Eylon},
  booktitle={Proceedings of the Fourteenth Annual ACM-SIAM Symposium on Discrete Algorithms {(SODA)}},
  pages={1096--115},
  year={2020},
  organization={SIAM}
}

@article{nevsetvril2016distributed,
  title={A distributed low tree-depth decomposition algorithm for bounded expansion classes},
  author={Ne{\v{s}}et{\v{r}}il, Jaroslav and De Mendez, Patrice Ossona},
  journal={Distributed Computing},
  volume={29},
  number={1},
  pages={39--49},
  year={2016},
  publisher={Springer}
}

@book{nevsetvril2012sparsity,
  title={Sparsity: graphs, structures, and algorithms},
  author={Ne{\v{s}}et{\v{r}}il, Jaroslav and De Mendez, Patrice Ossona},
  volume={28},
  year={2012},
  publisher={Springer Science \& Business Media}
}

@article{feuilloley2021compact,
  title={Compact Distributed Certification of Planar Graphs},
  author={Feuilloley, Laurent and Fraigniaud, Pierre and Montealegre, Pedro and Rapaport, Ivan and R{\'e}mila, {\'E}ric and Todinca, Ioan},
  journal={Algorithmica},
  volume={83},
  number={7},
  pages={2215--2244},
  year={2021},
  publisher={Springer US New York}
}

@article{dmtcs:8479,
  TITLE = {{Introduction to local certification}},
  AUTHOR = {Laurent Feuilloley},
  URL = {https://dmtcs.episciences.org/8479},
  DOI = {10.46298/dmtcs.6280},
  JOURNAL = {{Discrete Mathematics \& Theoretical Computer Science}},
  VOLUME = {{vol. 23, no. 3}},
  YEAR = {2021},
  MONTH = Sep,
}

@article{lovasz2006graph,
  title={Graph minor theory},
  author={Lov{\'a}sz, L{\'a}szl{\'o}},
  journal={Bulletin of the American Mathematical Society},
  volume={43},
  number={1},
  pages={75--86},
  year={2006}
}

@article{kawarabayashi2007some,
  title={Some recent progress and applications in graph minor theory},
  author={Kawarabayashi, Ken-ichi and Mohar, Bojan},
  journal={Graphs and combinatorics},
  volume={23},
  number={1},
  pages={1--46},
  year={2007},
  publisher={Springer}
}

@article{arnborg1991easy,
  title={Easy problems for tree-decomposable graphs},
  author={Arnborg, Stefan and Lagergren, Jens and Seese, Detlef},
  journal={Journal of Algorithms},
  volume={12},
  number={2},
  pages={308--340},
  year={1991},
  publisher={Elsevier}
}

@article{ROBERTSON2004325,
title = {Graph Minors. {XX.} {Wagner}'s conjecture},
journal = {Journal of Combinatorial Theory, Series B},
volume = {92},
number = {2},
pages = {325-357},
year = {2004},
note = {Special Issue Dedicated to Professor W.T. Tutte},
doi = {https://doi.org/10.1016/j.jctb.2004.08.001},
author = {Neil Robertson and Paul D Seymour},
}

@article{robertson2003graph,
  title={Graph minors. {XVI.} Excluding a non-planar graph},
  author={Robertson, Neil and Seymour, Paul D},
  journal={Journal of Combinatorial Theory, Series B},
  volume={89},
  number={1},
  pages={43--76},
  year={2003},
  publisher={Elsevier}
}

@inproceedings{montealegre2021compact,
  title={Compact distributed interactive proofs for the recognition of cographs and distance-hereditary graphs},
  author={Montealegre, Pedro and Ram{\'\i}rez-Romero, Diego and Rapaport, Ivan},
  booktitle={International Symposium on Stabilizing, Safety, and Security of Distributed Systems {(SSS)}},
  pages={395--409},
  year={2021},
  organization={Springer}
}

@inproceedings{jauregui2022distributed,
  title={Distributed interactive proofs for the recognition of some geometric intersection graph classes},
  author={Jauregui, Benjamin and Montealegre, Pedro and Rapaport, Ivan},
  booktitle={International Colloquium on Structural Information and Communication Complexity {(SIROCCO)}},
  pages={212--233},
  year={2022},
  organization={Springer}
}

@inproceedings{kol2018interactive,
  title={Interactive distributed proofs},
  author={Kol, Gillat and Oshman, Rotem and Saxena, Raghuvansh R},
  booktitle={Proceedings of the 2018 ACM Symposium on Principles of Distributed Computing {(PODC)}},
  pages={255--264},
  year={2018}
}

@article{THOMASON2001318,
title = {The Extremal Function for Complete Minors},
journal = {Journal of Combinatorial Theory, Series B},
volume = {81},
number = {2},
pages = {318--338},
year = {2001},
author = {Andrew Thomason},
}

@article{elek2022planarity,
  title={Planarity can be verified by an approximate proof labeling scheme in constant-time},
  author={Elek, G{\'a}bor},
  journal={Journal of Combinatorial Theory, Series A},
  volume={191},
  pages={105643},
  year={2022},
  publisher={Elsevier}
}

@article{feuilloley2020bibliography,
  title={Bibliography of distributed approximation beyond bounded degree},
  author={Feuilloley, Laurent},
  journal={arXiv preprint arXiv:2001.08510},
  year={2020}
}

@inproceedings{ghaffari2016distributed,
  title={Distributed algorithms for planar networks {II}: Low-congestion shortcuts, {MST}, and min-cut},
  author={Ghaffari, Mohsen and Haeupler, Bernhard},
  booktitle={Proceedings of the twenty-seventh annual ACM-SIAM symposium on Discrete algorithms {(SODA)}},
  pages={202--219},
  year={2016},
  organization={SIAM}
}

@inproceedings{ghaffari2021low,
  title={Low-congestion shortcuts for graphs excluding dense minors},
  author={Ghaffari, Mohsen and Haeupler, Bernhard},
  booktitle={Proceedings of the 2021 ACM Symposium on Principles of Distributed Computing {(PODC)}},
  pages={213--221},
  year={2021}
}

@inproceedings{chang2023efficient,
  title={Efficient Distributed Decomposition and Routing Algorithms in Minor-Free Networks and Their Applications},
  author={Chang, Yi-Jun},
  booktitle={Proceedings of the 2023 ACM Symposium on Principles of Distributed Computing {(PODC)}},
  pages={55--66},
  year={2023}
}

@inproceedings{chang2022narrowing,
  title={Narrowing the {LOCAL}--{CONGEST} Gaps in Sparse Networks via Expander Decompositions},
  author={Chang, Yi-Jun and Su, Hsin-Hao},
  booktitle={Proceedings of the 2022 ACM Symposium on Principles of Distributed Computing {(PODC)}},
  pages={301--312},
  year={2022}
}

@inproceedings{IzumiSPAA22,
author = {Izumi, Taisuke and Kitamura, Naoki and Naruse, Takamasa and Schwartzman, Gregory},
title = {Fully Polynomial-Time Distributed Computation in Low-Treewidth Graphs},
year = {2022},
isbn = {9781450391467},
publisher = {Association for Computing Machinery},
address = {New York, NY, USA},
url = {https://doi.org/10.1145/3490148.3538590},
doi = {10.1145/3490148.3538590},
booktitle = {Proceedings of the 34th ACM Symposium on Parallelism in Algorithms and Architectures {(SPAA)}},
pages = {11--22},
}

@inproceedings{ghaffari2016planar,
  title={Distributed algorithms for planar networks {I}: Planar embedding},
  author={Ghaffari, Mohsen and Haeupler, Bernhard},
  booktitle={Proceedings of the 2016 ACM Symposium on Principles of Distributed Computing {(PODC)}},
  pages={29--38},
  year={2016}
}

@article{levi2021property,
  title={Property testing of planarity in the CONGEST model},
  author={Levi, Reut and Medina, Moti and Ron, Dana},
  journal={Distributed Computing},
  volume={34},
  number={1},
  pages={15--32},
  year={2021},
  publisher={Springer}
}

@article{bodlaender1998partial,
  title={A partial $k$-arboretum of graphs with bounded treewidth},
  author={Bodlaender, Hans L},
  journal={Theoretical computer science},
  volume={209},
  number={1-2},
  pages={1--45},
  year={1998},
  publisher={Elsevier}
}

@inproceedings{bodlaender1989nc,
  title={{NC}-algorithms for graphs with small treewidth},
  author={Bodlaender, Hans L},
  booktitle={proceedings of the 14th International Workshop on Graph-Theoretic Concepts in Computer Science (WG 1988)},
  pages={1--10},
  year={1989},
  organization={Springer}
}

@inproceedings{afek1991memory,
  title={Memory-efficient self stabilizing protocols for general networks},
  author={Afek, Yehuda and Kutten, Shay and Yung, Moti},
  booktitle={Distributed Algorithms: 4th International Workshop {(WDAG 1990)}},
  pages={15--28},
  year={1991},
  organization={Springer}
}

@inproceedings{haeupler2016low,
  title={Low-congestion shortcuts without embedding},
  author={Haeupler, Bernhard and Izumi, Taisuke and Zuzic, Goran},
  booktitle={Proceedings of the 2016 ACM Symposium on Principles of Distributed Computing {(PODC)}},
  pages={451--460},
  year={2016}
}

@inproceedings{haeupler2016near,
  title={Near-optimal low-congestion shortcuts on bounded parameter graphs},
  author={Haeupler, Bernhard and Izumi, Taisuke and Zuzic, Goran},
  booktitle={Proceedings of the International Symposium on Distributed Computing {(DISC)}},
  pages={158--172},
  year={2016},
  organization={Springer}
}

@inproceedings{haeupler2018minor,
  title={Minor excluded network families admit fast distributed algorithms},
  author={Haeupler, Bernhard and Li, Jason and Zuzic, Goran},
  booktitle={Proceedings of the 2018 ACM Symposium on Principles of Distributed Computing {(PODC)}},
  pages={465--474},
  year={2018}
}

\end{document}